\documentclass[12pt, draftclsnofoot, onecolumn]{IEEEtran}
\usepackage{amsmath,amssymb,cite}
\usepackage{mathtools}
\usepackage[small]{caption}
\usepackage{amsthm,multirow,color,amsfonts}
\usepackage{tabulary}
\usepackage{subfigure}
\usepackage{graphicx}
\usepackage{setspace}
\usepackage{enumerate}	
\usepackage{comment}

\makeatletter
\g@addto@macro{\normalsize}{%
	\setlength{\abovedisplayskip}{3pt plus 2pt minus 2pt}
	\setlength{\abovedisplayshortskip}{3pt plus 2pt minus 2pt}
	\setlength{\belowdisplayskip}{3pt plus 2pt minus 2pt}
	\setlength{\belowdisplayshortskip}{3pt plus 2pt minus 2pt}
	\setlength{\textfloatsep}{10pt plus 2pt minus 2pt}
}
\makeatother

\title{Throughput Maximization for Delay-Sensitive Random Access Communication}
\author{Derya~Malak,~Howard~Huang, and~Jeffrey~G.~Andrews
\thanks{Part of the manuscript has been presented at the 2017 IEEE ICC Wireless Communications Symposium \cite{MaHuAnd2017ICC}.}
\thanks{D. Malak is with the Research Laboratory of Electronics (RLE), at The Massachusetts Institute of Technology, Cambridge, MA 02139 USA (email: deryam@mit.edu).} 
\thanks{H. Huang is with Bell Labs, Nokia, Murray Hill, NJ 07974, USA (e-mail: howard.huang@nokia-bell-labs.com).}
\thanks{J. G. Andrews is with the Wireless Networking and Communications Group (WNCG), The University of Texas at Austin, Austin, TX 78712, USA (email: jandrews@ece.utexas.edu).} 
\thanks{This work was done when D. Malak was with WNCG, The University of Texas at Austin. 
\hfill Last revised: {\today}.}}

\newtheorem{prop}{Proposition}

\newtheorem{cor}{Corollary}

\newtheorem{lem}{Lemma}

\newcommand{\fbl}{\mathrm{FBL}}      
\newcommand{\snr}{\mathrm{SNR}}    
\newcommand{\bfsnr}{\mathrm{{\bf SNR}}}   
\newcommand{\sinr}{\mathrm{SINR}}       
  
\newcommand{\SINR}{{\sf{SINR}}}    
\newcommand{\SINRone}{{\sf{SINR}}_1}
\newcommand{\SINRm}{{\sf{SINR}}_m}
\newcommand{\SINRM}{{\sf{SINR}}_M}
   
\newcommand{\sinrone}{{\sf{\zeta}}_1}     
  
\newcommand{\sinri}{{\sf{\zeta}}_i}
\newcommand{\sinrm}{{\sf{\zeta}}_m}  
\newcommand{\sinrM}{{\sf{\zeta}}_M}      
\newcommand{\Cfbl}{C}

\newcommand{\Pfail}{{\mathrm{P}_{\sf Fail}}}

\newcommand{\PfailIBL}{{\mathrm{P}_{\sf Fail, IBL}}}
\newcommand{\PfailFBL}{{\mathrm{P}_{\sf Fail, FBL}}}

\DeclareMathOperator\Q{{Q}}

\newcommand{\Bopt}{B_{\sf opt}}
\newcommand{\Mopt}{M_{\sf opt}}
\newcommand{\Lopt}{\lambda_{\sf opt}}
\newcommand{\kopt}{k^*}
\newcommand{\Gopt}{\Gamma_{\sf opt}}
\newcommand{\eber}{{\epsilon_{\sf FBL}}}

\newcommand{\ICI}{{\sf ICI}}
\newcommand{\GF}{\Gamma_{\rm o}}
\begin{document}

\maketitle
\begin{abstract} 

Future 5G cellular networks supporting ultra-reliable, low-latency communications (URLLC) could employ random access communication to reduce the overhead compared to scheduled access techniques used in 4G networks. We consider a wireless communication system where multiple devices transmit payloads of a given fixed size in a random access fashion over shared radio resources to a common receiver. We allow retransmissions and assume Chase combining at the receiver. The radio resources are partitioned in the time and frequency dimensions, and we determine the optimal partition granularity to maximize throughput, subject to given constraints on latency and outage.  In the regime of high and low signal-to-noise ratio (SNR), we derive explicit expressions for the granularity and throughput, first using a Shannon capacity approximation and then using finite block length analysis. Numerical results show that the throughput scaling results are applicable over a range of SNRs. The proposed analytical framework can provide insights for resource allocation strategies in general random access systems and in specific 5G use cases for massive URLLC uplink access. 
\end{abstract}


\maketitle

\section{Introduction}
\label{intro} 
Ultra-reliable, low-latency communication (URLLC) for a large number of devices will be an important use case for future 5G communication networks  \cite{DhiHuVis2015}. This so-called machine-type communication (MTC) is often characterized by small payload sizes transmitted by a very large number of devices, and is not well served by the 4G LTE air interface which is designed to support larger payloads transmitted by fewer devices. 

For multiple devices communicating on the uplink to a base, LTE uses the following (simplified) three-phase procedure. In the first phase,  devices wanting to transmit notify the base station on an uplink random access channel. 
In the second phase, the base station decides how to best allocate resources among the devices and notifies them through a downlink message. In the third phase, the devices transmit to the base over the assigned uplink resources. The first two phases can be thought of as overhead in ensuring interference-free transmission during the scheduled access. If the payload size is sufficiently large, then the overhead is justified. On the other hand, for small payloads found in MTC, the overhead is not justified, and it is worthwhile to consider transmitting the payload itself in a random access fashion \cite{Dhillon2014tcomm, Laya2014}. For a given set of spectral resources, the general challenge is to design a random access protocol to maximize throughput subject to quality of service constraints.  

In this paper, we use a Poisson arrival process to model devices entering the system, so the throughput can be characterized by the mean arrival rate of this process which we denote as $\lambda$. Each device attempts to transmit a payload of fixed size $L$ bits over shared resources of bandwidth $W$ Hertz which is partitioned into $B$ frequency bins and into time slots. We refer to a given frequency bin and time slot as a time-frequency slot (TFS). For a given transmission attempt, a device transmits during a slot over a randomly chosen bin. Depending on the arrival rate, multiple devices could transmit over the same TFS.  If a given transmission is not successfully decoded (as a result of the signal-to-interference plus noise ratio (SINR) falling below a threshold), the device can retransmit its signal. For a given payload, a device can transmit up to $M$ times within a given time constraint of $T$ seconds.  At the receiver, Chase combining is used to combine the received energy over multiple transmission attempts. Chase combining is a common form of hybrid automatic repeat request (HARQ) which has been shown to increase throughput in relatively poor channel conditions and has been standardized in 3G and 4G cellular standards. If a device is not able to have its packet decoded within $M$ transmission attempts (or equivalently, within time $T$ seconds), then its transmission is a failure. The overall problem is to determine the optimal values of $B$ and $M$ to maximize the Poisson mean arrival rate $\lambda$, subject to a constraint on the probability of failure which we denote as $\delta$.  

Our analysis therefore characterizes tradeoffs among throughput $\lambda$, reliability $\delta$, latency $T$, and bandwidth $W$, using a practical receiver which decodes a desired device's transmission by treating interference as noise. One could use our framework to extract insights on the throughput performance of uplink 5G URLLC, for example, by coupling a low failure probability (e.g., $\delta = 0.001$) with a low latency requirement (e.g., $T = 0.005$ seconds). 


\subsection{Related Work and Motivation}
\label{relatedwork}

A number of recent references describe different strategies for accommodating massive uplink access including algorithms for collision resolution \cite{Madueno2014}, spatial diversity  with multiple base antennas \cite{Johansson2015},  load control via pricing algorithms  \cite{Koseoglu2016}, and interference cancellation  \cite{Dovelos2017}. Because of possibly limited energy resources of devices in these use cases, another body of work considers both throughput and energy efficiency as metrics \cite{Biral2015},\cite{MalakDhillonAndrews2016}, \cite{Dhillon2013twc}. 

In terms of more fundamental results, reference \cite{YuITA2017} provides a degrees-of-freedom characterization of throughput when considering both device identification and communication. 
In \cite{Dhillon2013twc}, an optimal transmitter and receiver strategy for maximizing the number of devices transmitting data with a fixed payload size is derived. The optimal receiver jointly decodes a subset of the devices using an interference canceller, where the subset is determined randomly based on the target outage rate. 

Because ideal interference cancellation is not realizable in practice, especially for a large number of devices, reference \cite{MaHuAnd2017ICC} characterizes the throughput of a suboptimal but more practical random access system where both the time and frequency domains are slotted. The receiver uses conventional single-user detection which demodulates a desired user's data stream by treating other users' interfering signals as noise, and as a simplifying assumption, the analysis uses Shannon capacity to approximate the SINR threshold for a failed transmission. This approximation leads to an optimistic bound on the throughput, and it becomes exact in the limit of infinite coding block lengths. 

The current paper reviews and extends the results in \cite{MaHuAnd2017ICC} by incorporating recent characterizations of capacity under finite block length transmissions \cite{Polyanskiy2010,YanCaiDurPol2015, Durisi2016tcomm}. This extension is especially relevant for MTC applications where the payload size could be as small as a few hundred bits. A related random access framework for finite block length transmissions is  discussed in \cite{Durisi2015}, and a similar framework for downlink URLLC is studied in \cite{AnaVec2018}.

\subsection{Overview and contributions}

The goal of this paper is to analyze a general framework for uplink random access communications to provide insights on resource allocation strategies relevant for future 5G networks. As described in the system model in Section \ref{RACHmodel}, the time and frequency resources are partitioned into slots whose granularity is optimized in order to maximize the throughput, given constraints on latency and reliability. In Section \ref{IBL} we develop general closed-form expressions for the reliability in terms of the probability of transmission failure under two power control scenarios, one in which all devices are received with exactly the same SNR ({\em constant SNR} scenario) and another in which average received power is the same but the realization is modulated by fading ({\em Rayleigh fading} scenario). These results use an approximation of the SINR threshold related to failure, assuming capacity-achieving encoding with infinite block length (IBL) transmissions. In Section \ref{results},  explicit expressions for the optimal slot granularity are derived by focusing on the regime of high and low $\snr$, and the scaling of the optimal throughput is characterized for each regime.  In Section \ref{FBL}, we exploit the scaling results for the IBL regime to characterize the optimal throughput for finite block length (FBL) transmissions, again under low and high $\snr$ conditions. We numerical results in Section \ref{Simulations} and show that the throughput scaling results are applicable for even moderate SNRs.

The key design insights for the proposed random access framework are as follows:
\begin{itemize}
\item For the constant $\snr$ scenario in the high $\snr$ regime, the resources should be split sufficiently such that devices do not experience any interference provided that the total number of resources is sufficiently large. Hence, the number of frequency bins $B$ should 
scale with the number of device arrivals $K$.
\item For the constant $\snr$ scenario in the low $\snr$ regime, we show that the devices should share the resources. 
\item For the case of Rayleigh fading, although the variability of the channel causes a drop in the number of arrivals that can successfully complete the random access phase, we briefly discuss that similar conclusions extend to that case. 
\item FBL regime is the practical case for short packet sizes. Although there is a gap in its throughput 
from the IBL model, scaling results for both regimes are similar.
\end{itemize}
 
Within the framework of a suboptimal single-user receiver, our analysis is an upper bound on the achievable throughput. This is due to the following additional assumptions: ideal negative acknowledgement with no error or delay, IBL or FBL capacity-achieving encoding, perfect power control, and perfect synchronization among users. Relaxing these assumptions could be studied in the future through a more general analysis or through simulations.  Nevertheless, the design insights could be applied to 5G cellular system design where delay-constrained communications will be an important use case.

\section{System Model and Assumptions}
\label{RACHmodel} 

We consider a wireless multiple access communication channel where a set of users transmit over $W$ Hz bandwidth of shared radio resources to a single receiver. Each user attempts to communicate a payload of $L$ bits within a latency constraint of $T$ seconds. For a given payload, a user can in general transmit multiple times within $T$ seconds according to a retransmission protocol we describe below. For each user, the attempts to transmit a payload are modeled as a Poisson process with a common mean arrival rate. If we further assume the devices transmit independently, then the aggregate transmissions of the users can be modeled as a Poisson process with mean arrival rate $\lambda$ given by the sum of the individual mean arrival rates. 

The time domain of the resources is partitioned into slots of equal duration, which we define below. During a given slot, a random number of new users enter the system and attempt to transmit, and  unsuccessful users from previous slots may also attempt to retransmit. The frequency domain is partitioned into $B$ bins of bandwidth $W/B$ Hz each. So on a given time slot, each user (new or returning) chooses one of the $B$ bins at random and sends a sequence of encoded symbols over the entire time slot and frequency bin. We refer to a given slot and bin as a time frequency slot (TFS).  Note that given $W$, $T$, $B$, and $M$, the number of symbols in a TFS is $WT/(BM)$. The coding rate is therefore $LBM/(WT)$ bits per transmitted symbol. 


Each user experiences Rayleigh fading with unit mean. Then for a given user $j$ the channel fading power $h_j$ follows an exponential distribution with mean $1$ (identical distribution for all $j$). Power control is employed so that the average received power at the base station from any user is the same. The noise is assumed to be additive Gaussian with constant power. Therefore as a result of the Rayleigh fading, the signal-to-noise ratio (SNR), defined as the ratio of the user's received power to the noise power, is $\rho h_j$. 

In general, a given TFS chosen by a user could be occupied by other users. Suppose a total of $K$ users have chosen a TFS and their squared channel amplitudes are $h_1,\ldots,h_K$. Then the signal-to-interference plus noise ratio (SINR) of a desired user's power, say with index 1, with respect to the others' is $\rho h_1 / (1 + \rho \sum_{k=2}^K h_j)$.  A given user's first transmission is successful if the SINR exceeds the required information theoretic threshold $\Gamma$. We derive $\Gamma$ later and see that it is different for the IBL and FBL cases.  If it is unsuccessful (i.e., the SINR is less than $\Gamma$), then a negative acknowledgement is sent to the transmitting user to indicate a failed transmission, and the user has an opportunity to attempt retransmission. In this paper, we make an ideal assumption that this negative acknowledgement is sent with zero delay and error such that the user can immediately attempt a retransmission on the following time slot over a randomly selected frequency bin. 

Chase combining is used at the receiver to combine signal energy for a given user's transmissions over multiple slots.  The transmission is successful if the Chase-combined SINR exceeds the threshold $\Gamma$. Once a user's transmission is successfully decoded, it stops transmitting. Otherwise, it will continue to transmit on successive slots with a randomly selected bin, for up to a total of $M$ slots. The duration of each slot is $T/M$ seconds, so that if a user's transmission is successful by the $M^{\rm th}$ transmission, it will have met the latency constraint $T$ seconds. 

The overall problem is to determine the maximum Poisson arrival rate $\lambda$ that can be supported by adjusting the parameters $B$ and $M$, given the resource constraints of bandwidth $W$Hz bandwidth and time $T$ seconds, and for a given the user target payload size $L$ bits and outage probability $\delta$:
\begin{equation}
\label{optimization1}
\begin{aligned}
\Lopt=&\underset{B, \,M\in\mathbb{Z}^+}{\max}
& & \lambda \\
& \hspace{0.6cm}\text{s.t.} 
& & \Pfail(\lambda, L, B, M) \leq \delta, \\
\end{aligned}
\end{equation}
where  $\Pfail(\lambda, L, B, m)$ is the probability of failure for a typical user attempting to transmit a given payload, up to and including the $m^{\rm th}$ transmission $(m=1,\ldots,M)$. (We have suppressed the dependence on $W$ and $T$.) Given that each slot is $T/M$ seconds, $\Pfail(\lambda, L, B, M)$ is the probability that a user cannot meet its latency constraint $T$ seconds. If we denote a typical user's Chase-combined SINR on the $m$ transmission attempt as ${\sf{SINR}}_m$ ($m= 1,\ldots,M$), then a failure (or outage) event occurs if the Chase-combined SINR falls below the threshold $\Gamma$ for every attempt, up to and including the $m^{\rm th}$. Hence we can write: 
\begin{align}
\label{Poutdefinition}
\Pfail(\lambda,L,B,m)=\mathbb{P}\left[{\sf{SINR}}_1< \Gamma, {\sf{SINR}}_2<\Gamma,\ldots,{\sf{SINR}}_m<\Gamma \right].
\end{align}
 In general, the Chase-combined SINR ${\sf{SINR}}_m$  is dependent on $\lambda$, $L$, $B$, $W$, $T$, and the channel realizations of users choosing a particular TFS, but we suppress them to simplify the notation. For the first transmission $m=1$, there is no Chase combining, so that for a desired user with index 1 we have
 \begin{equation}
  {\sf{SINR}}_1  = \frac{\rho h_1}{ 1 + \rho \sum_{k=2}^K h_j},
 \end{equation}
as described above. The derivation the Chase-combined SINR for $m>1$ is given in Section \ref{IBL}. We note that the aggregate arrival rate for a given slot is the sum of the arrival rate of new users $\lambda$ and the arrival rate of users who are retransmitting. The retransmission rate is dependent on the failure probability, and we derive the aggregate arrival rate in Section \ref{IBL}.

As an example, consider the system model for $M=5$ retransmissions as shown in Fig. \ref{Mgeneral}. A device arriving during slot $1$ has until slot $5$ to transmit. Another device arriving during slot $2$ has until slot $6$ to transmit. This emphasizes the fact that devices can arrive during each slot. In the same figure, we also illustrate a typical device that uses only $m=3$ time slots out of $5$ to retransmit its payload to the BS by choosing frequency bins at random at each time slot. There are $k_i$ devices contending in frequency-time slot $i$. The first $2$ attempts combined at the BS cannot be decoded and the BS sends NACK back to the device after each failure. The payload is successfully decoded after $2$ retransmission attempts and the BS sends an ACK to the device. 

\begin{figure}[t!]
\center
\includegraphics[width=\columnwidth]{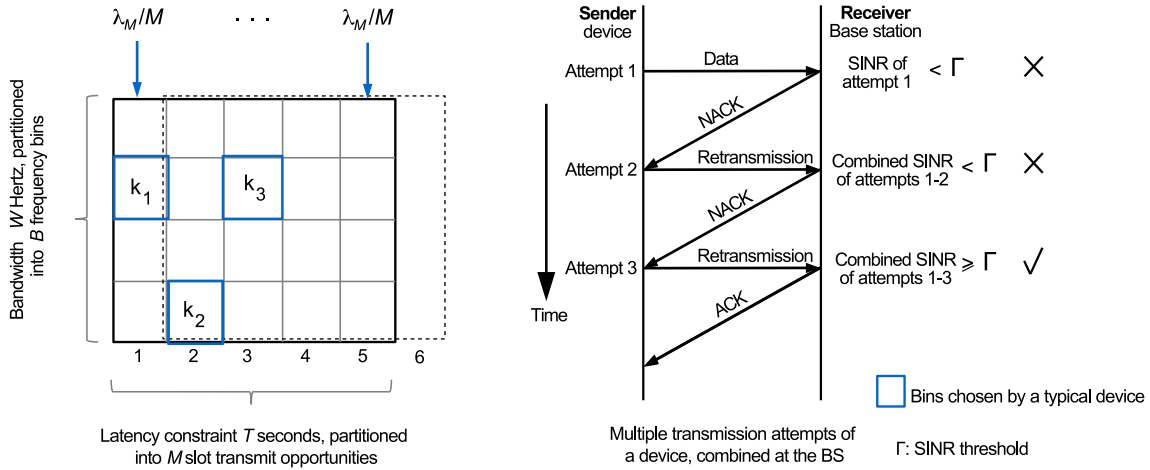} 
\caption{\small{An example RA scenario. The bandwidth is partitioned into $B (=4)$ frequency bins and time is partitioned into $M (=5)$ transmit opportunities. The average arrival rate of devices for each slot is $\lambda_M/M$. We illustrate a typical device that is unsuccessful on the first transmission attempt, and retransmits $2$ times by choosing frequency bins at random at each time slot. The transmission attempts are combined at the BS, yielding success in slot $m=3$.}} 
\label{Mgeneral}
\end{figure}

\section{Infinite Block Length}
\label{IBL}
Shannon's channel capacity is achievable at an arbitrarily low error rate when coding is performed in the infinite block length (IBL) regime, i.e., using a code block of infinite length. However, since the number of resources $N$ is finite, the ratio $L/N$ is always finite. Hence, given $L$, the IBL scheme gives an upper bound on the achievable rate, and a lower bound on $n$.

In this section, we optimistically assume that the encoded block length of $n$ symbols is sufficiently large that we can exploit the IBL\footnote{In Sect. \ref{FBL}, we focus on the finite block length (FBL) regime, where we no longer have the large block length assumption.} limit to characterize the system performance. The capacity of an additive white Gaussian noise (AWGN) channel where interference is treated as noise in that case is
\begin{align}
\label{shannoncap}
C(\SINR)=\log_2(1+\SINR). 
\end{align}

Denote by $\SINRm$ the Chase combiner output $\sinr$ up to and including the $m^{\rm th}$ attempt of the typical device. A device fails on the $m^{\rm th}$ attempt if $\SINRm$ is below the threshold $\Gamma$. Given a target $\sinr$ outage rate $\delta$ per device, outage occurs if the received $\sinr$, which is a function of the channel realizations and the number of devices sharing the same TFS resources up to the current realization, is below $\Gamma$. At any time slot $m\in\mathcal{M}=\{1,\hdots, M\}$, for successful decoding of a typical device in the uplink, the block length $n$ should be chosen sufficiently large so that the transmit rate $L/n$ is below the channel capacity:   
\begin{align}
\label{CapacityConstraintGeneralModel}
\frac{L}{n}\leq C(\SINRm),\quad m\in\mathcal{M},\quad n\leq \frac{TW}{MB}.
\end{align}
Note that when $M$ or $B$ are large, $n$ may not be made arbitrarily large. However, when $\lambda$ is high, the capacity constraint is binding and $n$ should be made sufficiently large, which will become more explicit in Sect. \ref{IBL}. Hence, the possible values of $n$ are jointly determined by $B$, $M$ and $\lambda_M$. We need to make sure that the capacity constraint (\ref{CapacityConstraintGeneralModel}) is satisfied by a proper choice of $n$.

Given a block length $n$, from (\ref{shannoncap}) and (\ref{CapacityConstraintGeneralModel}), the $\sinr$ threshold for the IBL regime is given as 
\begin{align}
\Gamma=2^{\frac{L}{n}}-1. 
\end{align}
Combining the capacity constraint in (\ref{CapacityConstraintGeneralModel}) with the Chase combiner output $\SINRm$, we investigate the probability of outage for both constant $\snr$ and Rayleigh fading. 

Given the number of frequency bins $B$, the number of arrivals per bin per time frame of duration $T$, is Poisson distributed with an average arrival rate of  $\lambda_B=\lambda/B$. Hence, the probability of $k\geq 1$ devices choosing a given resource bin out of $B$ bins for transmission, given that there is at least one arrival, is given by the following conditional probability: 
\begin{align}
\label{Dfunction}
D(k,\lambda_B)=\frac{\mathbb{P}(k \,\,\text{device arrivals per bin},\,\, k\geq 1)}{\mathbb{P}(k\geq 1)}=\frac{\lambda_B^k\exp(-\lambda_B)}{k!(1-\exp(-\lambda_B))},\quad k\geq 1,
\end{align}
where we note that the probability of at least one device arrival is $\mathbb{P}(k\geq 1)=1-\exp(-\lambda_B)$. 

Denote the average aggregate device arrival rate with up to $M$ total transmissions per frame by $\lambda_M$, which is the sum of the rates of the original arrivals per slot, $\lambda$, and the arrivals occurring as a result of failed retransmissions up to a maximum of $M-1$ consecutive attempts. It equals
\begin{align}
\label{lambdaMvslambda}
\lambda_M=\lambda\left[1+\sum\limits_{m=1}^{M-1}\Pfail(\lambda,L,B,m) \right],
\end{align}
where $\Pfail(\lambda,L,B,m)$ as given in (\ref{Poutdefinition}) is the fraction of devices that fail up to and including the $m^{\rm th}$ attempt. Hence, given a maximum number of transmission attempts $M$ per frame, the number of aggregate arrivals at any time slot $m$ is Poisson distributed\footnote{For tractability, we inherently have the Poisson distribution assumption for the composite arrival process. From \cite{Dhillon2014tcomm} and \cite{Abramson1970}, this assumption is justifiable when the number of retransmissions is not too large.} with an average rate of $\frac{\lambda_M}{M}$. 

For ease of notation, given $M$, denote by $k_m\geq 1$, $m\in \mathcal{M}$ the realization of the number of aggregate arrivals on the $m^{\rm th}$ attempt with respect to the typical device, i.e., the sum of the number of new arrivals during transmission $m$ and the number of retransmissions due to the previously failed attempts in TFS $m$. Denote by $\mathcal{K}_m\coloneqq \{k_1,\hdots, k_m\}$ the realizations for the number of aggregate arrivals up to and including $m^{\rm th}$ attempt. Hence, the set $\mathcal{K}_m$ consists of the realizations of a Poisson random variable with parameter $\frac{\lambda_M}{BM}$, due to (\ref{lambdaMvslambda}). We also denote by $\zeta_m$ as a function of $\mathcal{K}_m$, a realization of random variable $\SINRm$ whose distribution depends on the history of the arrivals.

\subsection{Constant $\snr$}\label{ConstantSNR}
At constant $\snr$, let $k_m$ arrivals choose a given frequency bin, each having $\snr$ $\rho$, during transmission $m\in \mathcal{M}$. Given the Chase combiner output $\sinr$ as a result of $m\in \mathcal{M}$ transmissions, i.e. $\sinrm$ as a function of $\mathcal{K}_m$, and incorporating (\ref{shannoncap}) in the capacity constraint in (\ref{CapacityConstraintGeneralModel}), the maximum number of aggregate arrivals that can be supported at time slot $m\in\mathcal{M}$ can be determined as a function of $\mathcal{K}_{m-1}$. We denote this maximum by $\kopt_m\coloneqq k_m(\mathcal{K}_{m-1})$.  

We next give the Chase combiner output as a result of $M$ total transmissions at constant $\snr$.
\begin{prop}
\label{chase}
{\bf Chase combiner output $\snr$.}
If $M>1$, a device is allowed to retransmit if the preceding one fails, for a total of $M$ transmissions. Given that there are $K_i\geq 1$ devices transmitting in TFS $i$, each having $\snr$ $\rho$, then the $\snr$ at the output of the Chase combiner by treating interference as noise (TIN) is
\begin{align}
\label{chaseSINR}
\SINRm=\frac{\rho m^2}{m+\rho \big(\sum\limits_{i=1}^m K_i-m\big)}, \quad m\in\mathcal{M}.   
\end{align}
\end{prop}

\begin{proof}
See Appendix \ref{App:Appendix-chase}.
\end{proof}

For a given maximum number of retransmission attempts $M$, an outage occurs if there are more than $\kopt_m$ devices sharing the same bin on attempt $m\in\mathcal{M}$, due to limited number of symbol resources $N/MB$. Using (\ref{Dfunction}), the probability of outage at slot $m\in\mathcal{M}$ is calculated adding up the probabilities $D(k_m,\lambda_B)$ over the number of arrivals that cannot be supported. Given $B$ and $N$, Prop. \ref{OutageGeneralM-nofading} gives the probability of outage at constant $\snr$ up to a maximum of $M$ transmissions.

\begin{prop}
\label{OutageGeneralM-nofading}
Given $B$, $M$, symbol resources $N$, the probability of outage in the IBL regime in the case of constant $\snr$ is given by the following expression:
\begin{align}
\label{PoutGeneralShannon}
\PfailIBL(\lambda,L,B,M)
=\sum\limits_{m=1}^M\,\,\sum\limits_{k_m=\kopt_m(\mathcal{K}_{m-1})+1}^{\infty}\,\,\prod\limits_{i=1}^M D\Big(k_i,\frac{\lambda_M}{B}\Big),
\end{align}
where the lower limit of the summation in (\ref{PoutGeneralShannon}) is
\begin{align}
\label{jmaxgeneralM}
\kopt_m(\mathcal{K}_{m-1})&=\max\Big\{\Big \lfloor m+\frac{m^2}{\Gamma}-\frac{m}{\rho}-\sum\limits_{i=1}^{m-1}k_i  \Big\rfloor, 1 \Big\},\,\,\, \Gamma=2^{\frac{L}{n}}-1,
 \,\,\, n\leq \frac{N}{BM}, \,\,\, m\in \mathcal{M}.
\end{align}
The aggregate arrival rate $\lambda_M$ with up to $M$ total transmissions is given by (\ref{lambdaMvslambda}).
\end{prop}

\begin{proof}
See Appendix \ref{App:Appendix-OutageGeneralM-nofading}.
\end{proof}

For successful transmission, given an $\sinr$ threshold $\Gamma$, we require that $\frac{L}{n}\leq C(\sinrm)$, $m\in \mathcal{M}$. Using (\ref{chaseSINR}), in order to satisfy the capacity constraint, we need to have $\sum\nolimits_{i=1}^m k_i \leq \big \lfloor m+\frac{m^2}{\Gamma}-\frac{m}{\rho}\big\rfloor,\, m\in \mathcal{M}$, yielding the lower limit of the summation in (\ref{PoutGeneralShannon}), as given by (\ref{jmaxgeneralM}).

\subsection{Small Scale (Rayleigh) Fading}\label{RayFad}
For the case of small scale Rayleigh fading with mean $1/\mu$, let $K_i$ arrivals choose a given frequency bin, each having $\snr$ $\rho$ at retransmission attempt $i$. Denote by $h_{l_i} \sim \exp(\mu)$ for $l\in\{1,\hdots, K_i\}$, $i\in\mathcal{M}$, the independent and identically distributed (i.i.d.) channel power distributions, where the desired device has index $l=1$, and $l\in\{2,\hdots, K_i\}$ is the interferer index. Let the total interference seen at transmission attempt $i$ be $I_{K_i}=\rho\sum\nolimits_{l=2}^{K_i}{h_{l_i}}$. Incorporating the channel power distributions, the Chase combiner output $\sinr$ as a result of $m\in\mathcal{M}$ transmissions, which is proven in Appendix \ref{App:Appendix-chase}, is 
\begin{align}
\label{SINRdefnRayleighFading}
\SINRm =\frac{\rho \big(\sum\limits_{i=1}^m h_{1_i}\big)^2}{\sum\limits_{i=1}^m h_{1_i}\big(1+I_{K_i}\big)}\overset{(a)}{=}\frac{m^2\rho h_{1_m}}{m+ \sum\limits_{i=1}^m I_{K_i}},\quad m\in\mathcal{M},
\end{align}
where $(a)$ is based on the assumption that $h_{1_m}$ is unchanged across $M$ time slots. This is consistent with the Rayleigh block fading model \cite{Zhong2016} in which the power fading coefficients remain static over each time slot, and are temporally (and spatially) independent. 

The probability of outage in the case of small scale Rayleigh fading is derived next.
\begin{prop}\label{ShannonOutageFadingGeneralM}
In the case of Rayleigh fading with mean $1/\mu$, the probability of outage as a function of $B, M, L, \rho, N$, and $\sinr$ threshold $\Gamma$ is characterized as
\begin{eqnarray}
\label{GeneralMFadingOutage}
\PfailIBL(\lambda,L,B,M)=\sum\limits_{m=1}^M\,\,\sum\limits_{l_m\in\{l_{m-1},\,l_{m-1}+1\}}(-1)^{l_M}\,\,\prod\limits_{i=1}^M f\Big(\mu,\frac{\lambda_M}{B},l_i\Gamma\Big),
\end{eqnarray}
where $l_0=0$ and the term $f(\mu,\alpha,\Gamma)$ is expressed as
\begin{align}
\label{f-function}
f(\mu,\alpha,\Gamma)=e^{-\mu \Gamma\rho^{-1}}(\Gamma+1)\frac{e^{\alpha/(\Gamma+1)}-1}{e^{\alpha}-1},
\end{align}
and the relation between $\Lopt$ and $\lambda_M$ is given by (\ref{lambdaMvslambda}).
\end{prop}

\begin{proof} 
See Appendix \ref{App:Appendix-ShannonOutageFadingM}.
\end{proof}

\begin{cor}
For the case of no retransmissions, the probability of outage can be derived as
\begin{align}
\label{M1FadingOutage}
\PfailIBL(\lambda,L,B,1)=1-f\Big(\mu,\frac{\lambda}{B},\Gamma\Big).
\end{align}
\end{cor}

We next investigate throughput scaling laws for different $\snr$ regimes exploiting the probability of outage given in (\ref{PoutGeneralShannon}). We only focus on the constant $\snr$ case, for which the analytical derivations are more tractable than the Rayleigh fading case. However, we do not present the technical details of the Rayleigh fading results since they complicate the analysis without providing any additional insights. We discussed in detail how to derive scaling results for the Rayleigh fading case exploiting the constant $\snr$ results in \cite{MaHuAnd2017ICC}. In Sect. \ref{Simulations}, we illustrate scaling results both for the constant $\snr$ and Rayleigh fading cases and provide numerical comparisons.

\section{Optimal Designs for High and Low $\snr$}
\label{results}
The goal of this section is to determine what are the optimal choices of the number of bins $B$, i.e., $\Bopt$,  and the number of retransmissions $M$, i.e., $\Mopt$, and the corresponding maximum average throughput 
$\Lopt$ for a fixed deadline constraint $T$. For the proposed random access setting, we solve the throughput optimization problem in (\ref{optimization1}) to maximize $\lambda$ with respect to the deadline constraint $T$ for high and low $\snr$ regimes, since this is not tractable for general $\snr$ settings. The total number of resources $N=TW$ is evenly split into $M$ retransmissions and $B$ bins. 
 
We investigate how $\Bopt$, $\Mopt$ and $\Lopt$ scale with the number of resource symbols $N$ at high and low $\snr$. Resource allocation is significantly different for these regimes. At one extreme, the combined $\sinr$ at the BS can be made high by splitting the resources into many resource bins and at the same time satisfying the outage and capacity constraints in (\ref{optimization1}). At another extreme, however, the combined $\sinr$ will be very low no matter how the resources are split, and the resources have to be shared in order not to violate the constraints of (\ref{optimization1}).

At constant $\snr$, the Chase combiner output $\sinr$ as a result of $m\in\mathcal{M}$ transmissions, i.e. $\sinrm$ as a function of $\mathcal{K}_m$, is derived from (\ref{chaseSINR}). Combining the capacity constraint in (\ref{optimization1}) with the Chase combiner output (\ref{chaseSINR}), and by incorporating (\ref{shannoncap}), we have
\begin{align}
\label{capacityconstraint}
\log_2(1+\sinrm)\geq \frac{L}{n} \geq \frac{LMB}{N},\quad m\in\mathcal{M},
\end{align}
which implies that $\sinrm\geq \Gamma=2^{\frac{L}{n}}-1$. This yields the following relation, 
\begin{align}
\label{GammaRhorelation}
m+\frac{m^2}{\Gamma}-\frac{m}{\rho}\geq \sum\limits_{i=1}^m k_i.
\end{align}
The maximum number of aggregate arrivals that can be supported in TFS $m\in\mathcal{M}$, denoted by $\kopt_m\coloneqq \kopt_m(\mathcal{K}_{m-1})$, is determined using (\ref{jmaxgeneralM}).

The constraint for the probability of outage in (\ref{optimization1}) as a function of $\mathcal{K}_m$ can be rewritten as
\begin{align}
\label{outageconstraint}
\mathbb{P}\Big[1+\frac{1}{\Gamma}-\frac{1}{\rho}<k_1,\, 1+\frac{2}{\Gamma}-\frac{1}{\rho}<\frac{1}{2}\sum\limits_{i=1}^2 k_i,\,\hdots, 1+\frac{M}{\Gamma}-\frac{1}{\rho}<\frac{1}{M}\sum\limits_{i=1}^M{k_i}\Big]\leq \delta.
\end{align}
Hence, a typical device fails when the number of aggregate arrivals at each retransmission attempt, i.e. $k_m$, $m\in\mathcal{M}$ given $\mathcal{K}_{m-1}$, exceeds some threshold.

The average rate of aggregate arrivals $\lambda_M$ will be approximately 
\begin{align}
\label{lambdaMapproximation}
\lambda_M\approx\Big(\sum\limits_{i=1}^M k_i\Big) B,
\end{align}
where $ B\leq \frac{N}{nM}$ and $n$ satisfies (\ref{capacityconstraint}). Therefore, we have the following relationship:
\begin{align}
\label{Bupperbound}
B\leq \frac{N}{LM}\log_2(1+\sinrm), \,\, m\in\mathcal{M}.
\end{align}

Combining (\ref{lambdaMapproximation}) and (\ref{Bupperbound}), the rate $\lambda_M$ will satisfy 
\begin{align}
\lambda_M\approx \Big(\sum\limits_{i=1}^M k_i\Big) \frac{N}{LM}\log_2\left(1+\sinrM\right)=\Big(\sum\limits_{i=1}^M k_i\Big) \frac{N}{LM}\log_2\left(1+\frac{\rho M^2}{M+\rho \big(\sum\nolimits_{i=1}^M k_i-M\big)}\right).\nonumber
\end{align}
Consider the simple case where $M=1$. Then, we have $\lambda\approx k_1 \frac{N}{L}\log_2\Big(1+\frac{\rho }{1+\rho (k_1-1)}\Big)$. For $M>1$, we can derive $\lambda$ from $\lambda_M$ using (\ref{lambdaMvslambda}). Now consider the following limiting cases of $\lambda_M$:
\begin{enumerate}[(i)]
\item Very high $\snr$ $\rho$:
\begin{align}
\lambda_M\approx\lim\limits_{\rho\to\infty} \Big(\sum\limits_{i=1}^M k_i\Big) \frac{N}{LM}\log_2\left(1+\sinrM\right)
=\Big(\sum\limits_{i=1}^M k_i\Big) \frac{N}{LM}\log_2\left(1+\frac{M^2}{\big(\sum\nolimits_{i=1}^M k_i-M\big)}\right),\nonumber
\end{align}
which is decreasing in $k_i\in\mathbb{Z}^+$.
\item Very low $\snr$ $\rho$: 
\begin{align}
\lambda_M\approx\lim\limits_{\rho\to 0}\Big(\sum\limits_{i=1}^M k_i\Big) \frac{N}{LM}\log_2\left(1+\sinrM\right)=0.\nonumber
\end{align}
\item Small number of devices $k_m$, $m\in\mathcal{M}$ per bin:
\begin{align}
\lambda_M\approx\lim\limits_{\underset{m\in\mathcal{M}}{k_m\to 1}}\Big(\sum\limits_{i=1}^M k_i\Big) \frac{N}{LM}\log_2\left(1+\sinrM\right)=\frac{N}{L}\log_2\Big(1+\rho M\Big).\nonumber
\end{align}
Furthermore, as $\rho\to 0$, $\frac{N}{L}\log_2\Big(1+\rho M\Big)\to\rho\frac{N}{L}\frac{M}{\log(2)}$, which decreases in $\snr$.
\item Large number of devices $k_m$, $m\in\mathcal{M}$ per bin:
\begin{align}
\lambda_M\approx\lim\limits_{\underset{m\in\mathcal{M}}{k_m\to \infty}}\Big(\sum\limits_{i=1}^M k_i\Big) \frac{N}{LM}\log_2\left(1+\sinrM\right)=\frac{N}{L}\frac{M}{\log(2)}.\nonumber
\end{align}
\end{enumerate}

At high $\snr$, it is clear from (i) that $k_i$, $i\in\mathcal{M}$ should be made very small. From (iii) and (iv), we have that $\lambda_M$ can be made larger by decreasing $k_i$. Besides, $k_i$ may not be increased arbitrarily since otherwise the outage constraint in (\ref{optimization1}) is violated. From (ii), we have that at low $\snr$, $\lambda_M$ will be very small. The relations (i)-(iv) indicate that at sufficiently high $\snr$, we have $k_i=1$, and at low $\snr$, $\rho\frac{N}{L}\frac{M}{\log(2)} \leq \lambda_M \leq \frac{N}{L}\frac{M}{\log(2)}$, where the lower limit is achieved when $k_i=1$ and upper limit is achieved as $k_i\to\infty$. Hence, at low $\snr$ $k_i$ should be made as large as possible. Its maximum can be attained by not splitting the frequency resources, i.e. $B=1$. 


\subsection{High $\snr$ Regime}
\label{ShannonHighSNR}
At high $\snr$, resources can be utilized more efficiently because we have more freedom in selecting $B$ and $M$ in the optimization problem (\ref{optimization1}). It is easy to satisfy the capacity constraint (\ref{capacityconstraint}) and the outage probability constraint (\ref{outageconstraint}), when the set of $k_m$'s, $m\in\mathcal{M}$, are small. 

Observe from (\ref{jmaxgeneralM}) that given $k_m\geq 1$ arrivals choose a given frequency bin, each having $\snr$ $\rho$, the optimal 
number of aggregate arrivals $\kopt_m(\mathcal{K}_{m-1})$, $m\in\mathcal{M}$, can be made arbitrarily small by decreasing the block length $n$. 
When $\rho$ is high, from (\ref{GammaRhorelation}), $\Gamma=2^{\frac{L}{n}}-1$ should be high, and hence, $n\leq \frac{N}{MB}$ could be made 
small by increasing $B$ or $M$. However, $B$ or $M$ may not be increased arbitrarily because the capacity constraint of (\ref{optimization1}) may be violated for very small $n$.

Given the received $\snr$ $\rho$, we study the threshold $\Gamma_T$ such that $\kopt_m(\mathcal{K}_{m-1})=1$, $m\in \mathcal{M}$ whenever $\Gamma\geq\Gamma_T$, i.e. no other interferer is allowed in the same frequency bin. In this case, devices experience no interference and the outage constraint is eliminated from (\ref{optimization1}). Using (\ref{jmaxgeneralM}), when $\kopt_m(\mathcal{K}_{m-1})=1$, $m\in \mathcal{M}$, the relation between $\Gamma$ and $\rho$ 
is determined as 
\begin{align}
\kopt_m(\mathcal{K}_{m-1})=1
=\max\Big\{\Big \lfloor 1+\frac{m^2}{\Gamma}-\frac{m}{\rho}  \Big\rfloor, 1 \Big\}.\nonumber
\end{align}
Hence, for $\Gamma_T=\rho m$, the maximum number of devices per bin is $\kopt_m(\mathcal{K}_{m-1})=1$.

At high $\snr$, we have $k_m=1$ for $m\in\mathcal{M}$, and from (\ref{chaseSINR}) the Chase combiner output $\sinr$ linearly scales with $M$. Hence, from the capacity constraint in (\ref{capacityconstraint}) we require in the IBL regime that $\underset{m\in\mathcal{M}}{\min}\{C(\sinrm)\}=\underset{m\in\mathcal{M}}{\min} \log_2(1+\rho m)\geq \frac{L}{n}$.  
Equivalently, we need to satisfy $\rho\geq \Gamma=2^{\frac{L}{n}}-1$ where $n\leq N/(MB)$. This yields the 
bound $B\leq \frac{N}{ML}\log_2(1+\rho)$. 

The outage probability in (\ref{PoutGeneralShannon}) can be made arbitrarily small, and the constraint $\PfailIBL(\lambda,L,B,M)\leq \delta$ is satisfied by choosing $B$ sufficiently large so that there is no other interferer in the same resource bin. To maximize $\lambda$, the number of bins $B$ should be maximized given the resource constraints.

\begin{prop}\label{Bopt_generalM_nofading-HighSNR}
Given $\delta, L, \rho, N$, the optimal number of bins, $\Bopt$, at high $\snr$ is
\begin{eqnarray}
\label{OptimalB-highSNR-nofading}
\Bopt=\frac{N}{\Mopt L}\log_2(1+\rho),
\end{eqnarray}
and the maximum 
user arrival rate per frame is $\lambda_{\Mopt}=\alpha_{\Mopt}(\delta) \Bopt$, where $\alpha_M(\delta)$ satisfies
\begin{align}
\label{alphadelta}
\alpha_M(\delta)=-\mathcal{W}\big((\delta^{\frac{1}{M}}-1)e^{\delta^{\frac{1}{M}}-1}\big)+\delta^{\frac{1}{M}}-1,
\end{align}
where $\mathcal{W}$ is the Lambert function and $\Mopt$ is given as
\begin{eqnarray}
\label{OptimalM-highSNR-nofading}
\Mopt=\underset{M \geq 1}{\arg\max}\,\,  B \alpha_M(\delta) (1-\delta^{1/M}),
\end{eqnarray}
where $B\approx \frac{N}{ML}\log_2(1+\rho)$.
\end{prop}
 
\begin{proof}
See Appendix \ref{App:Appendix-Bopt_generalM_nofading-HighSNR}.
\end{proof}

At high $\snr$, from Prop. \ref{Bopt_generalM_nofading-HighSNR}, we infer that $\Bopt\geq 1$ linearly scales with the number of resource symbols $N$. The number of devices should linearly scale with the number of frequency bins, and ideally such that there is a single user per bin. At high SNR, the number of bins will be higher than the number of users, i.e. $\alpha_M(\delta)<1$, when the target $\sinr$ outage rate $\delta$ is sufficiently small. Hence, resources should be split into as many frequency bins as possible such that the devices do not experience any interference. The throughput scales with the optimal number of bins $\Bopt$ and hence with $N$. The typical device will need fewer attempts as $\delta$ increases, since the $\sinr$ requirement for successful decoding decreases. The 
maximum number of (re)transmissions $\Mopt\geq 1$ is a non-increasing function of $\delta$, and can be determined from (\ref{OptimalM-highSNR-nofading}) as a function of $\delta$. 

\subsection{Low $\snr$ Regime}
\label{ShannonLowSNR}
At low $\snr$, since the devices mainly suffer from interference, and it might not be possible to split the symbol resources. It might also not be possible to improve the output $\sinr$ via Chase combining, and the Chase combiner output $\sinr$ will also be low. Due to low $\sinr$, the capacity constraint (\ref{capacityconstraint}) of the optimization formulation in (\ref{optimization1}) 
may only be satisfied when the symbol resources are shared, i.e., the set of $k_m$'s, $m\in\mathcal{M}$, are large. This can be handled by choosing a sufficiently large block length $n$. In addition, the target outage rate (\ref{outageconstraint}) can be met using a small threshold $\Gamma$. From (\ref{GammaRhorelation}), when $\rho$ is low, $\Gamma=2^{\frac{L}{n}}-1$ should be small, and hence, the block length $n\leq N/(MB)$ should be large (using small $B$, $M$). We can make $n$ as large as $N$ by letting $B=M=1$. At low $\snr$, we infer that all devices should share the resources in order to maximize the average arrival rate that can be supported.  


We next provide a sufficient condition for $\Bopt=\Mopt=1$.

\begin{prop}\label{lowSNRsufficientcondition}
Given the total number of resources $N$, when received $\snr$ $\rho$ is such that $\rho<2^{\frac{2L}{N}}-1$, then we have $\Bopt=\Mopt=1$.
\end{prop}

\begin{proof}
See Appendix \ref{App:Appendix-lowSNRsufficientcondition}. 
\end{proof}

From Prop. \ref{lowSNRsufficientcondition}, if $\snr$ is small enough, $\Bopt=\Mopt=1$, and since $\Gamma=2^{\frac{L}{n}}-1$ will be small, the outage constraint (\ref{outageconstraint}) will be met, and the resource utilization will be mainly constrained by the capacity constraint (\ref{capacityconstraint}). The throughput 
at low $\snr$ is determined next.

\begin{prop}\label{Lambda_approx_low_SNR}
At low $\snr$, we have $\Bopt=1$, $\Mopt=1$, and the maximum arrival rate $\Lopt$ can be approximated by
\begin{eqnarray}
\label{Lambdamax_approx_low_SNR}
\Lopt\approx \Big(\sqrt{\kopt_1+{(\Q^{-1}(\delta))^2}/{4}}-{\Q^{-1}(\delta)}/{2}\Big)^2,\quad \kopt_1=\Big\lfloor 1+\frac{1}{\Gamma}-\frac{1}{\rho}\Big\rfloor,
\end{eqnarray}
where $\Gamma=2^{\frac{L}{N}}-1$, and $\Q(x)$ is the tail probability of the standard normal distribution.
\end{prop}

\begin{proof}
See Appendix \ref{App:Appendix-Lambda_approx_low_SNR}.
\end{proof}


At low $\snr$, from Prop. \ref{Lambda_approx_low_SNR}, we can infer that all devices should share the resources in order to maximize the average arrival rate that can be supported. This is because the combined $\sinr$ at the BS will also be low, and for successful decoding the block length should be made as large as possible to decrease the threshold $\Gamma$. Hence, we require $\Bopt=\Mopt=1$. 

We summarize the scaling of parameters $\Lopt$, $\Bopt$ and $\Mopt$ in the IBL regime in Table \ref{table:tab3}.

It is possible to have a scenario in which the devices have different $\snr$ levels at the BS. To provide fairness among classes of devices with different $\snr$ levels, the fraction of  TFS resources allocated to each class can be optimized. However, we leave the resource optimization that supports different classes of devices as our future work.

\begin{table*}[t!]
\begin{center}
\setlength{\extrarowheight}{10pt}
\begin{tabular}{| c | c | }
\hline
{\bf Low $\snr$/Constant $\snr$}  &  {\bf Low $\snr$/Rayleigh fading}\\
$\Lopt\approx \big(\sqrt{\kopt_1+\frac{(\Q^{-1}(\delta))^2}{4}}-\frac{\Q^{-1}(\delta)}{2}\big)^2$ 
&  $\Lopt=\{\Lopt \vert\delta=1-f\big(\mu,\frac{\Lopt}{\Bopt},\Gopt\big)\}$ \\
$\Bopt=\Mopt=1^*$, $\kopt_1\approx \big \lfloor 1+\frac{1}{2^{\frac{L}{N}}-1}-\frac{1}{\rho}\big\rfloor$  & $\Bopt=\Mopt=1^*$, $\Gopt=2^{\frac{\Mopt L}{N}}-1$\\
\hline
{\bf High $\snr$/Constant $\snr$} & {\bf High $\snr$/Rayleigh fading}\\
$\Lopt=\lambda_{\Mopt}\Big(\frac{1-\delta^{1/{\Mopt}}}{1-\delta}\Big)$ & $\Lopt=\lambda_{\Mopt}/\Big[1+\sum\limits_{m=1}^{\Mopt-1}\PfailIBL(\lambda,L,B,m)\Big]$ \\
$\Mopt=\underset{M \geq 1}{\arg\max} \,\, \lambda_M (1-\delta^{1/M})$ & 
$\Mopt=\underset{M \geq 1}{\arg\max}\,\, \frac{\lambda_M}{1+\sum\limits_{m=1}^{M-1}m!(\mu\Gopt\rho^{-1})^m}$\\
$\Bopt=\frac{N}{\Mopt L}\log_2(1+\rho)$ & $\Bopt\approx\frac{N}{\Mopt L}\log_2(1+\Gopt)$, $\Gopt=\Mopt \log(\frac{1}{1-\delta})\rho$\\
\hline
\end{tabular}
\end{center}
\caption{Scaling results for different $\snr$ regimes at constant $\snr$, with Rayleigh fading \cite{MaHuAnd2017ICC} in the IBL regime.}
\label{table:tab3}
\end{table*}

\section{Finite Block Lengths}
\label{FBL}
In this section, we assume that the encoded block length $n$ symbols is finite, which is the practical case for short packet sizes. We exploit the finite block length (FBL) model of Polyanskiy {\em et al.} \cite{Polyanskiy2010} to characterize the performance of the FBL regime by optimizing the number of retransmission attempts and the required number of bins.

The maximal rate achievable with error probability $\varepsilon$, as a function of the block length $n$ and $\sinr$ is characterized by the normal approximation for Gaussian channels \cite{Polyanskiy2010}: 
\begin{align}
\label{FBLrate}
\Cfbl_{n,\varepsilon}(\SINR)= \log_2(1+\SINR)-\sqrt{\frac{V(\SINR)}{n}}\Q^{-1}(\varepsilon)+\frac{0.5\log_2(n)}{n}+o\left(\frac{1}{n}\right),
\end{align}
where $V(\SINR)$ is the channel dispersion given by 
\begin{align}
\label{dispersion}
V(\SINR)=\left[1-\frac{1}{(1+\SINR)^2}\right]\log_2^2(e).
\end{align}
Note that (\ref{FBLrate}) closely approximates $\Cfbl_{n,\varepsilon}(\SINR)$ for block lengths $n \geq 100$ bits \cite{Polyanskiy2010}. 

Given the payload size $L$ and encoded block length $n$, incorporating the capacity constraint (\ref{CapacityConstraintGeneralModel}) into (\ref{FBLrate}), the $\sinr$ threshold for the FBL regime is the inverse function of $\Cfbl_{n,\varepsilon}(\SINR)$ evaluated at $L/n$, i.e. $\Gamma=C_{n,\varepsilon}^{-1}(L/n)$. 
Combining the capacity constraint in (\ref{CapacityConstraintGeneralModel}) with the Chase combiner output $\SINRm$, we investigate the probability of outage at constant $\snr$.

Based on an asymptotic expansion of the maximum achievable rate in (\ref{FBLrate}), we calculate the block error rate for the FBL regime as  
\begin{align}
\label{BlockErrorRate}
\eber(n,L,\SINR)\approx \Q(f(n,L,\SINR)),
\end{align}
which is shown to be valid for packet sizes as small as $L = 50$ bits \cite{Polyanskiy2010}, where 
\begin{align}
f(n,L,\SINR)=\frac{n\log_2(1+\SINR)+0.5\log_2 n-L}{\sqrt{nV(\SINR)}}. 
\end{align}

The throughput optimization problem in the FBL regime can also be written as (\ref{optimization1}).  
We denote the probability of outage up to and including the $M^{\rm th}$ (re)transmission attempt for this regime by $\PfailFBL(\lambda,L,B,M)$, as will be detailed in Prop. \ref{FBLoutagenofadinggeneralM}. While multiple transmissions ($M>1$) cannot be independently decoded without error, i.e., the block error rate 
of an individual transmission can be larger than the target $\sinr$ outage rate $\delta$, Chase combining of $M$ transmissions may help meet the target outage rate. If $M=1$, the block error rate $\varepsilon$ has to satisfy $\varepsilon\leq\delta$. 

\begin{prop}\label{FBLoutagenofadinggeneralM}
Given $B$, $M$, and 
$N$, the probability of outage in the FBL regime is given by 
\begin{align}
\label{PoutdefinitionFBL}
\PfailFBL(\lambda,L,B,M)=\sum\limits_{m=1}^M\,\,\sum\limits_{k_m=1}^{\infty}\,\, \prod\limits_{i=1}^M D\Big(k_i,\frac{\lambda_M}{B}\Big)\eber\Big(n,L,\sinri\Big),
\end{align}
where
\begin{align}
\label{lambdaMvslambdaFBL}
\lambda_M=\lambda\left[1+\sum\limits_{m=1}^{M-1}\PfailFBL(\lambda,L,B,m)\right].
\end{align}
\end{prop}

\begin{proof}
See Appendix \ref{App:AppendixFBLoutagenofadinggeneralM}.
\end{proof}

Note that (\ref{lambdaMvslambdaFBL}) is a fixed point equation in the form of $\lambda_M=g(\lambda_M)$, where $g(\lambda_M)$ is determined by (\ref{PoutdefinitionFBL}). The solution $\lambda_M$ is unique because (\ref{PoutdefinitionFBL}) is a monotonically decreasing function of $\lambda_M$.

Note the relationship between the probability of outage for the FBL regime in (\ref{PoutdefinitionFBL}), and for the IBL regime, as given in (\ref{PoutGeneralShannon}). In the IBL regime, for a given block length $n$, for any $k_m$ value less than or equal to $\kopt_m(\mathcal{K}_{m-1})$ given in (\ref{jmaxgeneralM}), $m\in\mathcal{M}$, the capacity constraint in (\ref{CapacityConstraintGeneralModel}) is satisfied. However, in the FBL regime, the capacity constraint is stricter than the IBL capacity constraint. Given a block length $n$, although $\eber\big(n,L,\sinrm\big)$ can be made arbitrarily small for small $k_m$ at high $\snr$, $\eber\big(n,L,\sinrm\big)>0$ whenever $\sinrm$ is finite, even when $k_m=1$.

Although (\ref{FBLrate}) and (\ref{dispersion}) are valid for the additive white Gaussian noise (AWGN) channel \cite{Polyanskiy2010}, the block error probability changes for different retransmissions. Similar to \cite{Polyanskiy2010}, we assume that decoding errors are independent for different retransmissions. However, the block error probabilities are no longer independent because the Chase combiner output $\sinr$s across $m\in\mathcal{M}$ retransmissions, i.e. $\SINRm$, given in (\ref{chaseSINR}) and (\ref{SINRdefnRayleighFading}) for constant $\snr$ and Rayleigh fading, respectively, are dependent. Therefore, the combined $\sinr$ in (\ref{chaseSINR}), and the product form of $\PfailFBL(\lambda,L,B,M)$ in (\ref{PoutdefinitionFBL}) of Proposition \ref{FBLoutagenofadinggeneralM} are indeed approximations for the FBL regime.

Before we characterize the FBL regime at low and high $\snr$s, we give the following Lemma.

\begin{lem}\label{fconcaveofx}
For $X>0$, $f(n,L,X)$ is monotonically increasing, positive, and concave in $X$.
\end{lem}

\begin{proof}
See Appendix \ref{App:Appendixfconcaveofx}.\qedhere
\end{proof}

\begin{prop}\label{BoundsonQfunction}
Given a random variable $X$, $f(n,L,X)$ satisfies the following relation: 
\begin{align}
\label{JensenonQAndDecreasingQ}
\mathbb{E}[\Q(f(n,L,X))]\geq \Q(\mathbb{E}[f(n,L,X)]) \geq \Q(f(n,L,\mathbb{E}[X])).
\end{align}
\end{prop}

\begin{proof}
See Appendix \ref{App:AppendixBoundsonQfunction}.\qedhere
\end{proof}

We use Prop. \ref{BoundsonQfunction} in Sect. \ref{FBLlowSNR} to derive bounds for the probability of outage in the FBL regime.

The goal of the rest of this section is to figure out what the optimal choices of the number of bins $B$ and the number of retransmissions $M$ should be, and what this gives for the maximum average arrival rate $\lambda$ for the FBL regime. For the proposed random access setting, we solve the throughput optimization problem in (\ref{optimization1}) for the FBL regime to maximize $\lambda$ 
at high and low $\snr$ regimes, since this is not tractable for general $\snr$ settings as encountered in the IBL regime. The total number of resources $N=TW$ is evenly split into $M$ retransmissions and $B$ bins. 

Let $\Bopt^{\fbl}$, $\Mopt^{\fbl}$, $\Lopt^{\fbl}$ be the optimal number of bins, the optimal number of retransmissions, and the maximum throughput, 
respectively, for a fixed deadline $T$, in the FBL regime. We next investigate how $\Bopt^{\fbl}$, $\Mopt^{\fbl}$, $\Lopt^{\fbl}$ scale with the symbol resources 
$N$ at high and low $\snr$.

\subsection{High $\snr$ Regime}
\label{FBLhighSNR}
At high $\snr$, similar to the IBL regime, resource utilization can be optimized by splitting the resources, hence using multiple frequency bins, i.e. $\Bopt^{\fbl}>1$, and Chase combining to aggregate the received signals across multiple (re)transmissions, i.e. $\Mopt^{\fbl}>1$. We determine $\Bopt^{\fbl}$ and $\Mopt^{\fbl}$ by solving the throughput optimization problem for the FBL model in (\ref{optimization1}).

\begin{prop}\label{HighSNRFBLcoverage}
At high $\snr$, for a given device transmission, the probability of outage for the FBL regime can be approximated as
\begin{align}
\label{OutageNoiseLimited}
\PfailFBL(\lambda,L,B,M)\approx D\Big(1,\frac{\lambda_M}{B}\Big)^M\prod\limits_{m=1}^M \Q(f(n,L,\rho m)).
\end{align}
\end{prop}

\begin{proof}
See Appendix \ref{App:Appendix-HighSNRFBLcoverage}.
\end{proof}

We can observe that for the FBL regime, the achievable throughput of the multiuser channel will be lower than the IBL regime when the received $\snr$ per device at the BS is fixed, which is mainly due to the capacity constraint in (\ref{CapacityConstraintGeneralModel}) for the FBL regime. Hence, to achieve the same rate, it is required to increase the received $\snr$ $\rho$ per device at the BS.

\subsection{Low $\snr$ Regime}
\label{FBLlowSNR}
At low $\snr$, since the Chase combiner output $\sinr$ will be low, the capacity constraint in (\ref{CapacityConstraintGeneralModel}) may only be satisfied by choosing a sufficiently large block length $n$. Hence, the devices should share the resources such that $\Bopt^{\fbl}=\Mopt^{\fbl}=1$. 
Given the total number of resources $N$, since the FBL capacity constraint in (\ref{CapacityConstraintGeneralModel}) is stricter than the IBL capacity constraint 
of (\ref{optimization1}), when $\snr$ $\rho$ is such that $\rho<2^{\frac{2L}{N}}-1$ (see Sect. \ref{ShannonLowSNR}, Prop. \ref{lowSNRsufficientcondition}), it is guaranteed that 
$\Bopt^{\fbl}=\Mopt^{\fbl}=1$. 

Given $N$, Prop. \ref{LowSNRFBLcoverage} gives an approximation for the probability of outage for the FBL regime at high $\snr$ up to a maximum of $M$ transmissions.

\begin{prop}\label{LowSNRFBLcoverage}
At low $\snr$, for a given device transmission, the probability of outage for the FBL model can be approximated as
\begin{align}
\label{OutageLowSNRFBL}
\PfailFBL(\lambda,L,B,1)\approx\mathbb{E}_{\mathcal{K}}\big[\Q(f(n,L,\sinrone))\big],
\end{align}
where $\mathcal{K}$ is a random variable that denotes the number of arrivals on the $1^{\rm st}$ attempt given that there is at least one arrival. Its distribution is given by (\ref{Dfunction}).
\end{prop}

\begin{table*}[t!]
\begin{center}
\setlength{\extrarowheight}{10pt}
\begin{tabular}{| c | }
\hline
{\bf Low $\snr$/Constant $\snr$}  \\
$\Lopt^{\fbl}=\Big\{\lambda \Big\vert\delta=\sum\limits_{k\geq 1} D\big(k,\frac{\lambda}{B}\big)\eber(n,L,\sinrone)\Big\}$, $\Bopt^{\fbl}=\Mopt^{\fbl}=1^*$ \\
\hline
{\bf High $\snr$/Constant $\snr$} \\
$\Lopt^{\fbl}=\lambda_{\Mopt}^{\fbl}/\Big[1+\sum\limits_{m=1}^{\Mopt^{\fbl}-1}\PfailFBL(\lambda,L,B,m)\Big]$ \\
$\Mopt^{\fbl}=\underset{M \geq 1}{\arg\max}\,\, \lambda_M^{\fbl}/\Big[1+\sum\limits_{m=1}^{M-1}\PfailFBL(\lambda,L,B,m)\Big]$ \\
$\Bopt^{\fbl}=\Big\{  B\Big\vert \max\limits_{B\vert n=\frac{N}{B{\Mopt^{\fbl}}}} \, \bigcap\limits_{m=1}^{\Mopt^{\fbl}} {\frac{L}{n} \leq \Cfbl_{\big(\frac{N}{B m}\big),\big(\delta^{1/{\Mopt^{\fbl}}}\big)}\left(m\rho\right)} \Big\}$ \\
\hline
\end{tabular}
\end{center}
\caption{Scaling results for constant $\snr$ in the FBL regime.}
\label{table:tab4}
\end{table*}

\begin{proof}
At low $\snr$, for the FBL regime, the probability of outage can be derived as
\begin{align}
\PfailFBL(\lambda,L,B,1)=\sum\limits_{k=1}^{\infty}D\Big(k,\frac{\lambda}{B}\Big)\eber(n,L,\sinrone)=\mathbb{E}_{\mathcal{K}}\big[\eber(n,L,\SINRone)\big],\nonumber
\end{align}
where the final result follows from the approximation in (\ref{BlockErrorRate}).
\end{proof}

We next provide a lower and upper bound on the probability of outage at low $\snr$ for FBL.
 
\begin{prop}\label{LowSNRFBLcoverageBounds}
At low $\snr$, for a given device transmission, the probability of outage for the FBL model can be lower and upper bounded as follows
\begin{align}
\label{OutageInterferenceLimited}
\Q(f(n,L,\mathbb{E}_{\mathcal{K}}[\SINRone]))\leq\PfailFBL(\lambda,L,B,1)
\leq\sum\limits_{k=1}^{\infty}D\Big(k,\frac{\lambda}{B}\Big)\exp{(-f(n,L,\sinrone)^2)}.
\end{align}
\end{prop}

\begin{proof}
See Appendix \ref{App:Appendix-LowSNRFBLcoverageBounds}.
\end{proof}

At low $\snr$, similar to the IBL regime, all devices should share the resources to maximize the throughput. 
We illustrate the scaling of parameters $\Lopt^{\fbl}$, $\Bopt^{\fbl}$ and $\Mopt^{\fbl}$ for constant $\snr$ at high and low $\snr$ in the FBL regime in Table \ref{table:tab4}. These results can be extended to characterize the scaling of throughput versus latency for Rayleigh fading case, which is left as future work.

\section{Numerical Simulations}
\label{Simulations}
We investigate the scaling of the throughput $\Lopt$ in (\ref{optimization1}), with respect to $W$, $\rho$ and $T$ for constant $\snr$ and Rayleigh fading cases. For numerical simulations, we use fixed-point iteration to determine the relation between $\lambda_M$ and $\lambda$ given in (\ref{lambdaMvslambda}). In the plots, solid and dashed curves denote the analytical approximations for the scaling results, while numerical results are squares.

\begin{figure*}[t!]
\centering
\begin{minipage}[t]{.49\textwidth}
\centering
\includegraphics[width=\textwidth]{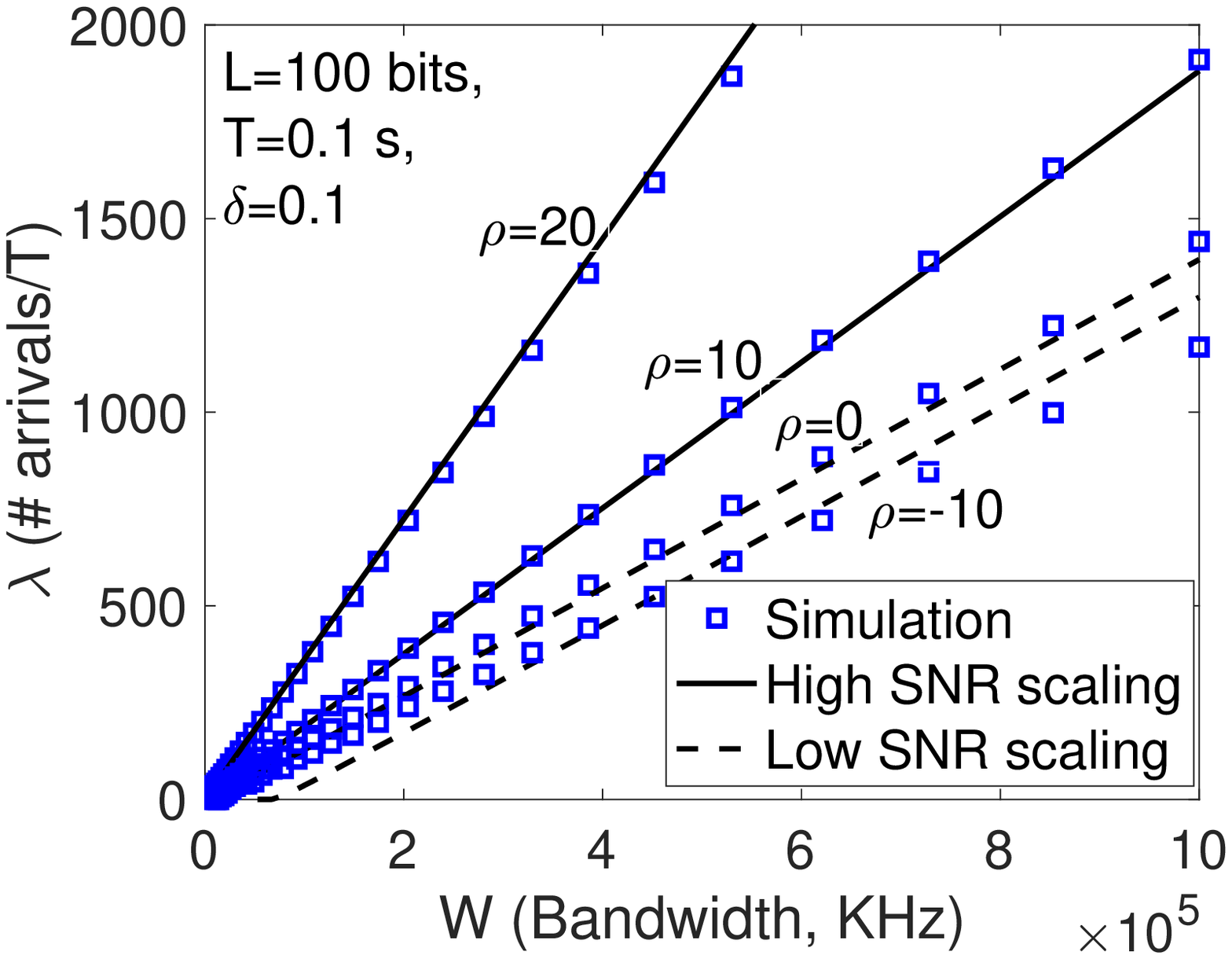}\\
\caption{\small{$\Lopt$ vs. bandwidth $W$ (kHz) for different \\$\rho$ (dB) with constant $\snr$, in the IBL regime.}\label{ThroughputvsBandwidth}}
\end{minipage}
\begin{minipage}[t]{.49\textwidth}
\centering
\includegraphics[width=\columnwidth]{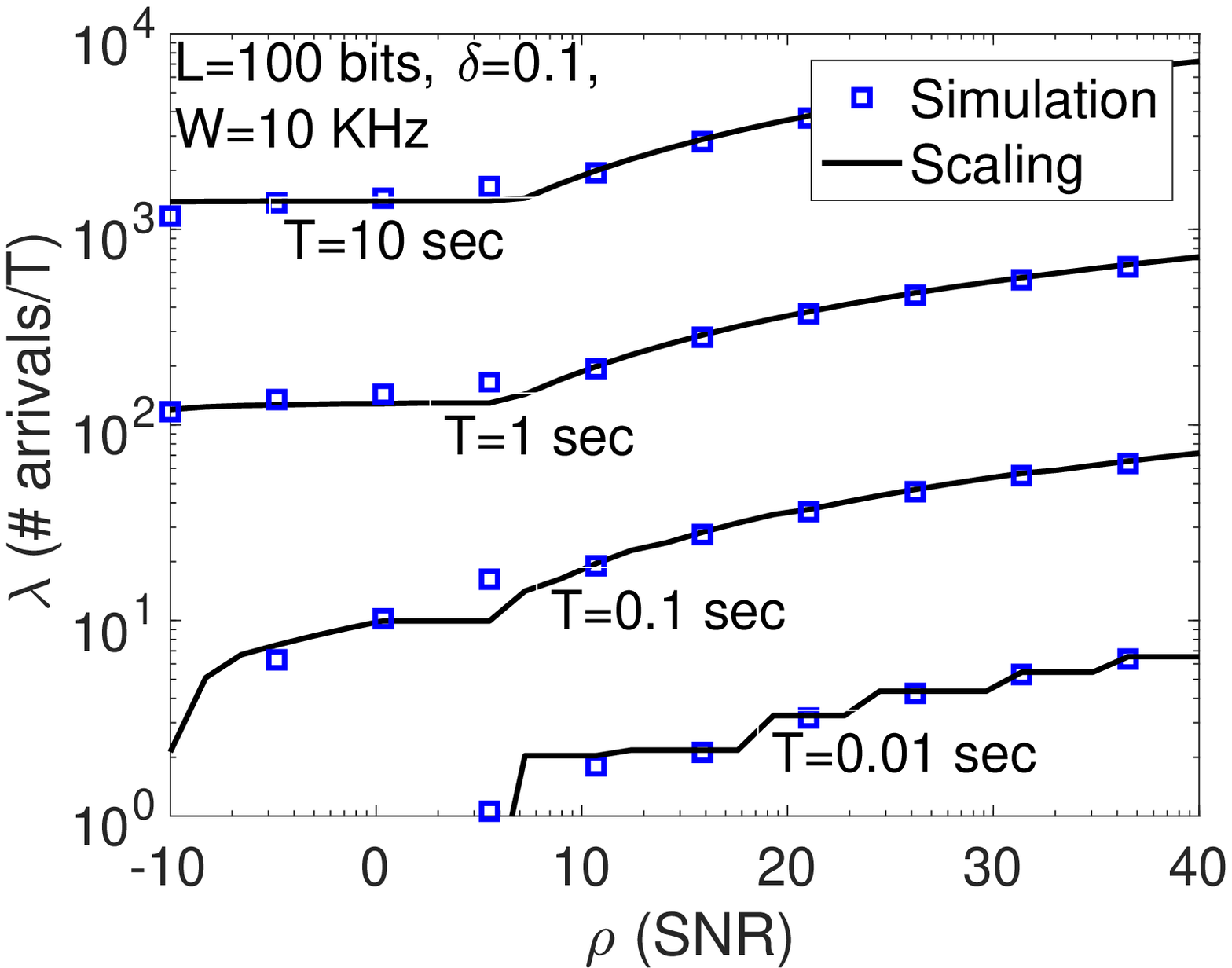}
\caption{\small{$\Lopt$ vs. $\snr$ $\rho$ (dB) for different delay constraints $T$ (s) with constant $\snr$, in the IBL regime.}\label{ThroughputvsRho}}
\end{minipage}
\end{figure*}

{\bf Throughput scaling with respect to bandwidth $W$.} At high $\snr$, there is a linear relationship between $\Lopt$ and $W$ because from Prop. \ref{Bopt_generalM_nofading-HighSNR}, $\Lopt/B$ is fixed for given $\delta$ and $B$ linearly scales with $W$. Given the fixed parameters $L$, $T$, and $\delta$, the slope $\Lopt/W$ at high $\snr$ as a function of $\rho$ can be found using Table \ref{table:tab3}. At high $\snr$, the capacity can be approximated as $C(\SINR)\approx 0.332 \cdot \SINR \,(\text{dB})$, where $\SINR \,(\text{dB})=10\log_{10}(\SINR)$. This yields
\begin{align}
\label{highSNRthroughputapprox}
\frac{\Lopt}{W}
\approx 0.332\alpha_{\Mopt}(\delta)\frac{T}{\Mopt L}\Big(\frac{1-\delta^{1/\Mopt}}{1-\delta}\Big)\cdot\rho\,(\text{dB}),
\end{align}
where $\rho \,(\text{dB})=10\log_{10}(\rho)$. 
At low $\snr$, since $\Bopt=\Mopt=1$, using $\Gamma=2^{\frac{L}{N}}-1\approx \frac{L}{N}\log(2)$, we obtain $\kopt_1\approx \lfloor 1+\frac{N}{L \log(2)}-\frac{1}{\rho}\rfloor$. From this approximation and (\ref{Lambdamax_approx_low_SNR}), we observe that $\Lopt$ scales sublinearly in $W$. Given $\rho$, 
the minimum amount of increase in $W$ that yields a change in the value of $\kopt_1$ hence $\Lopt$ is $\Delta W=\frac{L}{T}\log(2)$. Hence, as $T$ gets smaller or $L$ gets larger, it is not possible to achieve a linear scaling between $\Lopt$ and $W$. However, at low $\snr$, if $\delta$ is sufficiently large, 
$\Lopt\approx \kopt_1$. The trend of $\Lopt$ versus $W$ is illustrated in Fig. \ref{ThroughputvsBandwidth}.

\begin{figure*}[t!]
\centering
\begin{minipage}[t]{.49\textwidth}
\centering
\includegraphics[width=\textwidth]{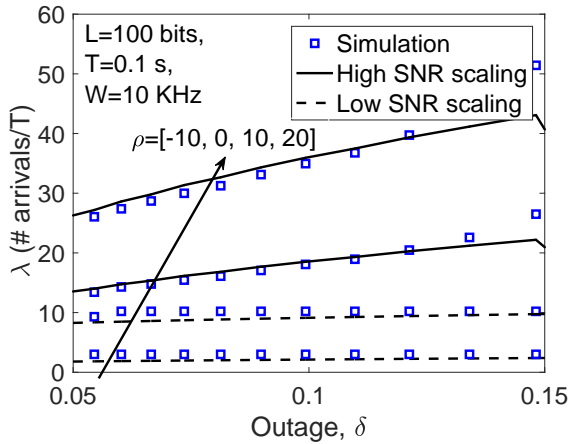}
\caption{\small{$\Lopt$ vs. outage constraint $\delta$ for different $\rho$ (dB).}}
\label{ThroughputvsOutage}
\end{minipage}
\begin{minipage}[t]{.49\textwidth}
\centering
\includegraphics[width=\textwidth]{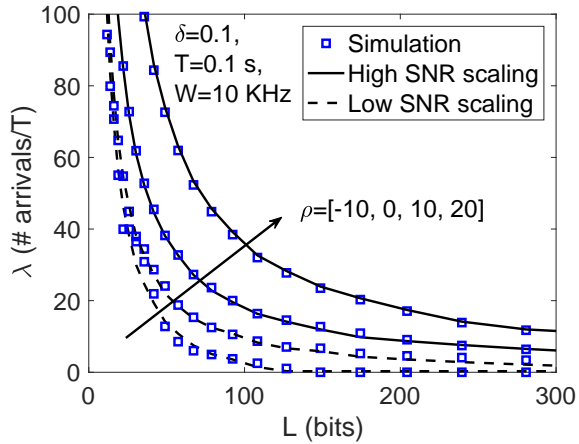}
\caption{\small{$\Lopt$ vs. payload $L$ (bits) for different $\rho$ (dB).}\label{ThroughputvsPayload}}
\end{minipage}
\end{figure*}

{\bf Throughput scaling with respect to $\bfsnr$ $\rho$.} Using the high $\snr$ approximation, and the expression for $\Lopt$ in Table \ref{table:tab3}, we obtain (\ref{highSNRthroughputapprox}). For different latency constraints ranging between $T=10$ msec and $T=10$ sec, we illustrate the trend of $\Lopt$ versus $\rho$ in Fig. \ref{ThroughputvsRho}.

{\bf Throughput scaling with respect to target outage rate $\delta$.} We observe the tradeoff between the throughput and latency from (\ref{optimization1}) such that $T$ is lower bounded as $T\geq \frac{LBM}{WC(\SINRM)}$. Note that $\Lopt$ increases with the target outage rate $\delta$. Large $\delta$ implies a rate constrained system and there are a lot of devices per bin. Small $\delta$ implies an $\sinr$ constrained model with fewer number of devices per bin. In Fig. \ref{ThroughputvsOutage}, we plot the trend of $\Lopt$ with respect to $\delta$ for different $\rho$ (dB). At low $\snr$, $\Bopt=\Mopt=1$, and $\Lopt$ increases with $\kopt$. However, $\kopt$ is bounded by the target outage rate in (\ref{optimization1}) (see Table \ref{table:tab3}), using which $\Lopt$ can be obtained. At high $\snr$, $\Bopt, \Mopt\geq1$, and $\kopt=1$, $\Lopt$ changes linearly in $B$, and is computed using the rate constraint in (\ref{optimization1}).

{\bf Throughput scaling with respect to packet size $L$.}
In Fig. \ref{ThroughputvsPayload}, we illustrate the trend with respect to $L$ for different $\rho$ (dB), $\delta=0.1$, deadline $T=0.1$ sec, and bandwidth $W=10$ KHz. It is intuitive that the maximum average supportable rate $\Lopt$ decays in payload size $L$. At low $\snr$, since $\Gamma$ increases in $L$, to satisfy a given outage constraint in (\ref{optimization1}), $\Lopt$ decays in $L$. Similarly, at high $\snr$, from Table \ref{table:tab3}, $\Bopt$ is inversely proportional to $L$ and from
\begin{align}
\Lopt\approx\alpha_{\Mopt}(\delta)\Bopt{(1-\delta^{1/\Mopt})}\big/{(1-\delta)},\nonumber
\end{align}
which justifies the inversely proportional relation between $\Lopt$ and $L$. We can also observe that $\Bopt$ and $\Mopt$ decrease in $L$ due to the capacity constraint in (\ref{optimization1}).

\begin{figure*}[t!]
\centering
\begin{minipage}[t]{.49\textwidth}
\centering
\includegraphics[width=\textwidth]{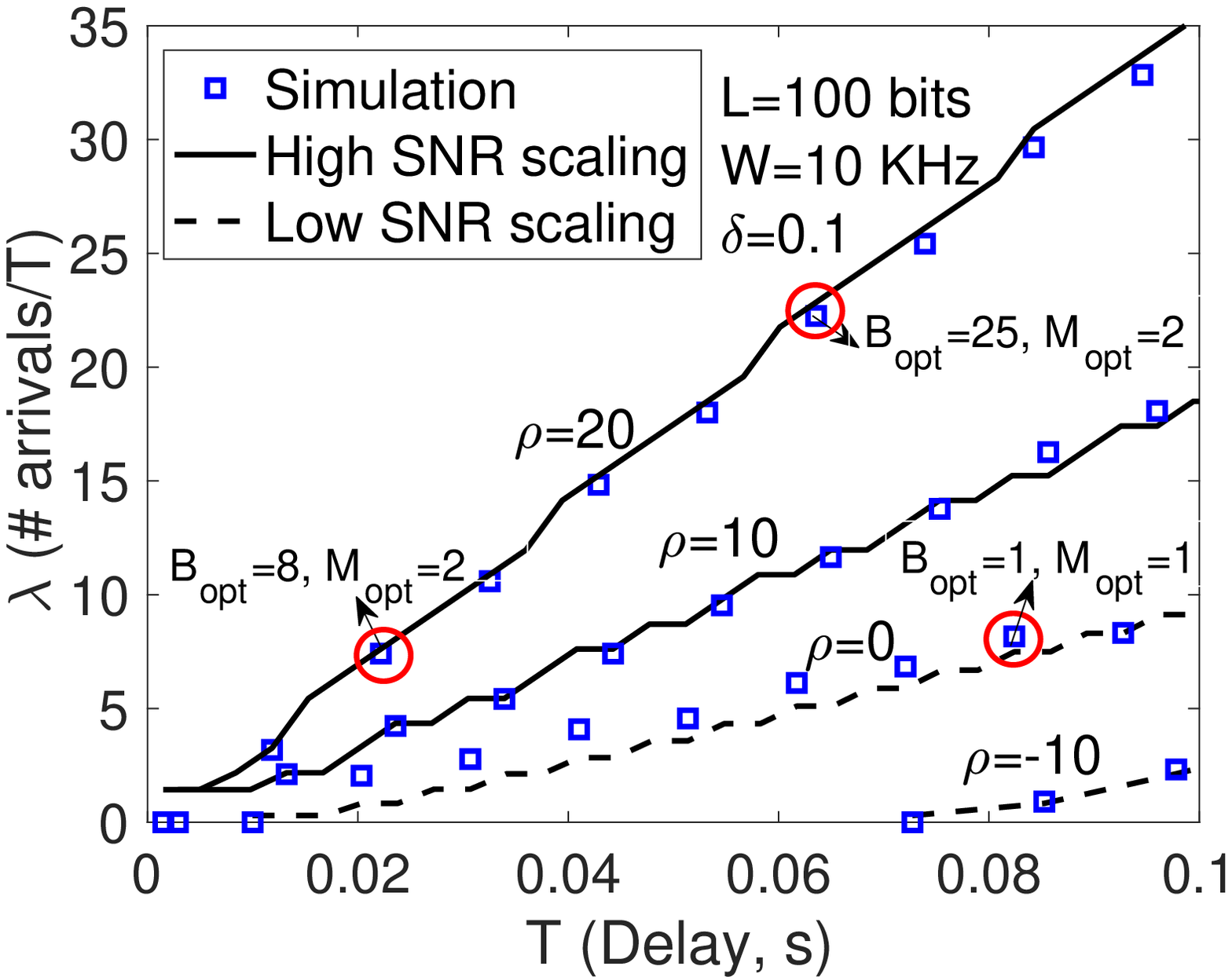}\\
\small{(a) Constant $\snr$.}
\end{minipage}
\begin{minipage}[t]{.49\textwidth}
\centering
\includegraphics[width=\textwidth]{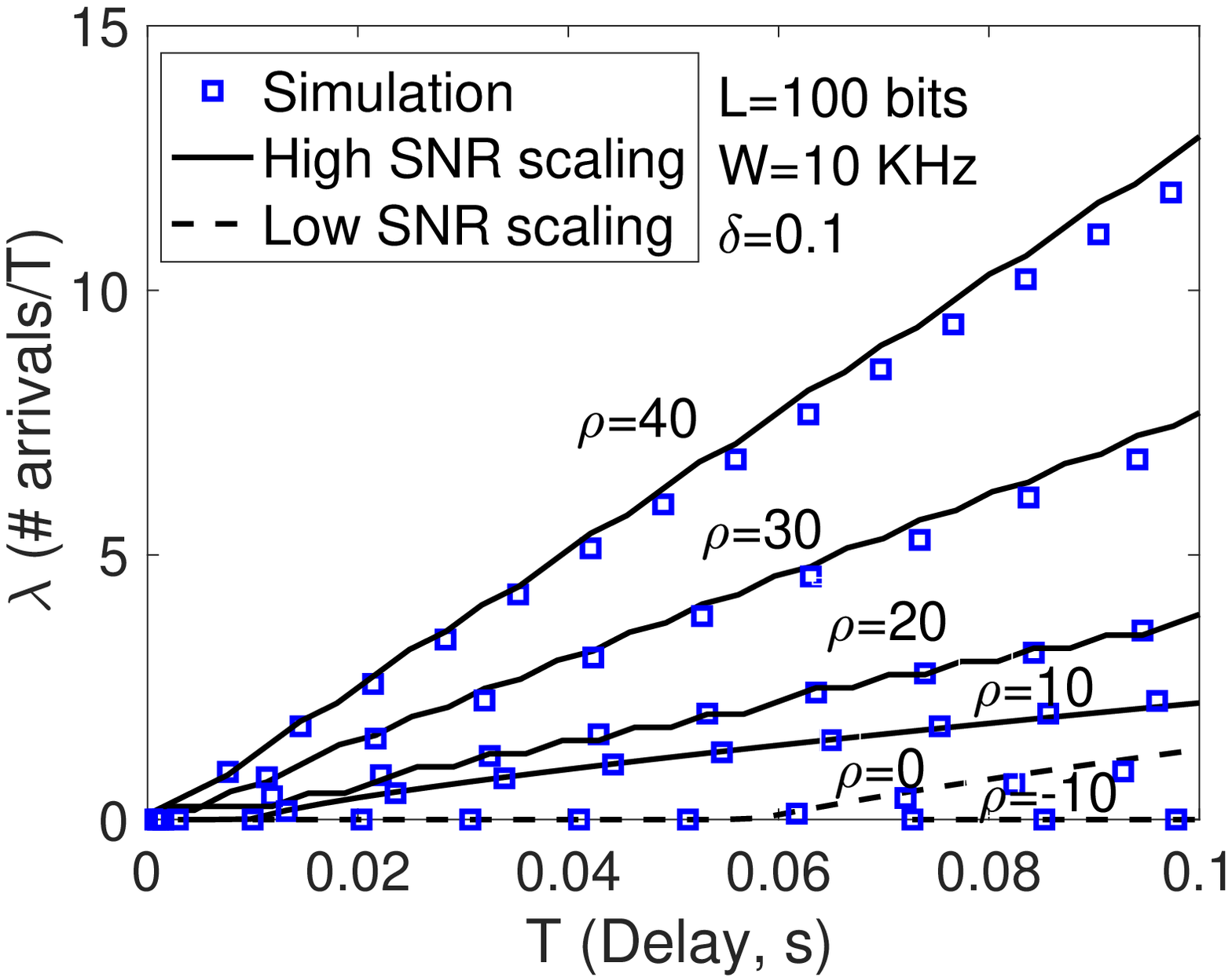}
\small{(b) Rayleigh fading with $\mu=1$ \cite{MaHuAnd2017ICC}.}
\end{minipage}
\caption{\small{$\Lopt$ vs. delay constraint $T$ (s) for different $\rho$ (dB) in the IBL regime.}\label{ThroughputvsDelay}}
\end{figure*}

\begin{figure*}[t!]
\centering
\begin{minipage}[t]{.49\textwidth}
\centering
\includegraphics[width=\textwidth]{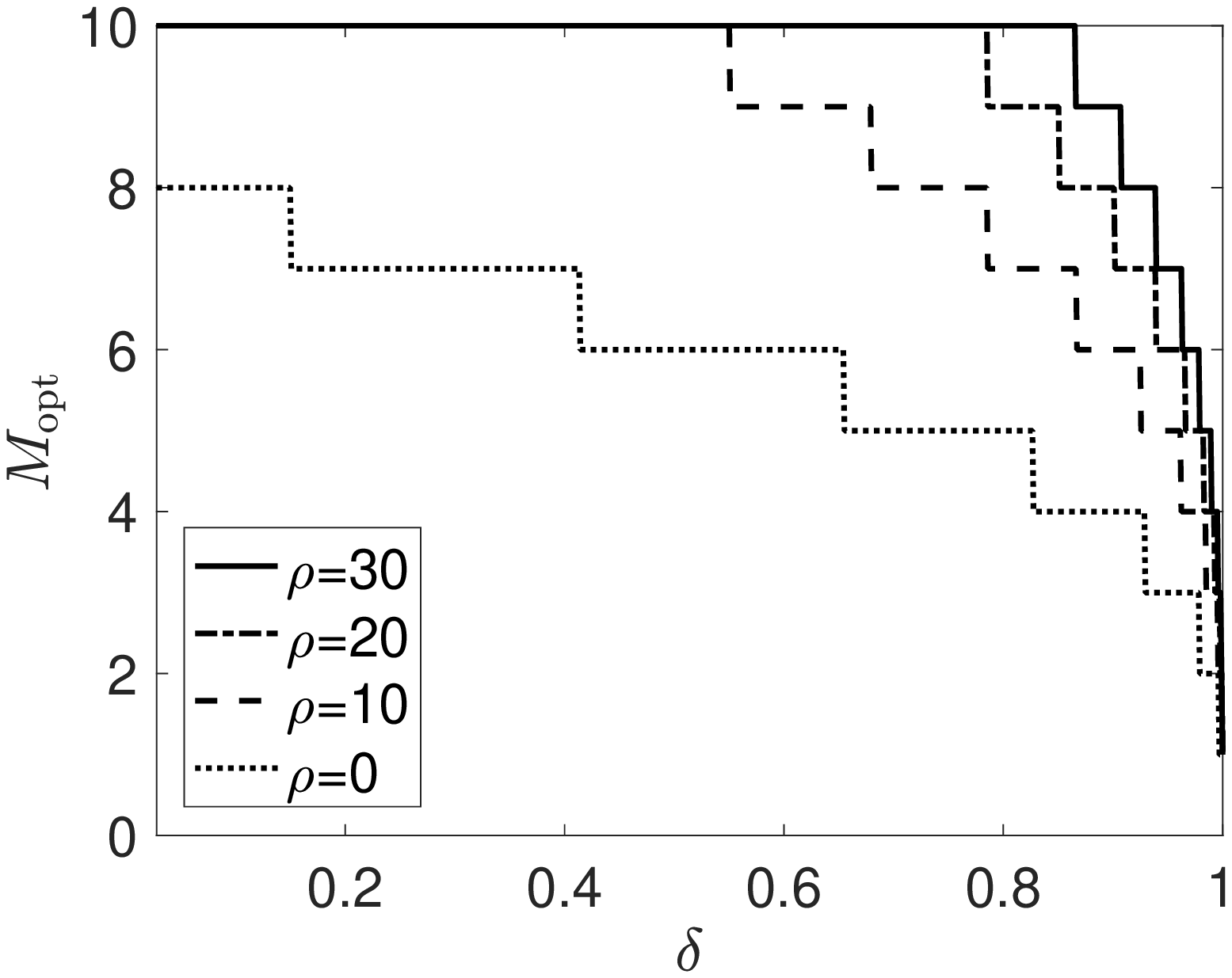}\\
\caption{\small{$\Mopt$ versus $\delta$ for different $\rho$ (dB) at high SNR.}\label{Moptvsdelta}}
\end{minipage}
\begin{minipage}[t]{.49\textwidth}
\centering
\includegraphics[width=\textwidth]{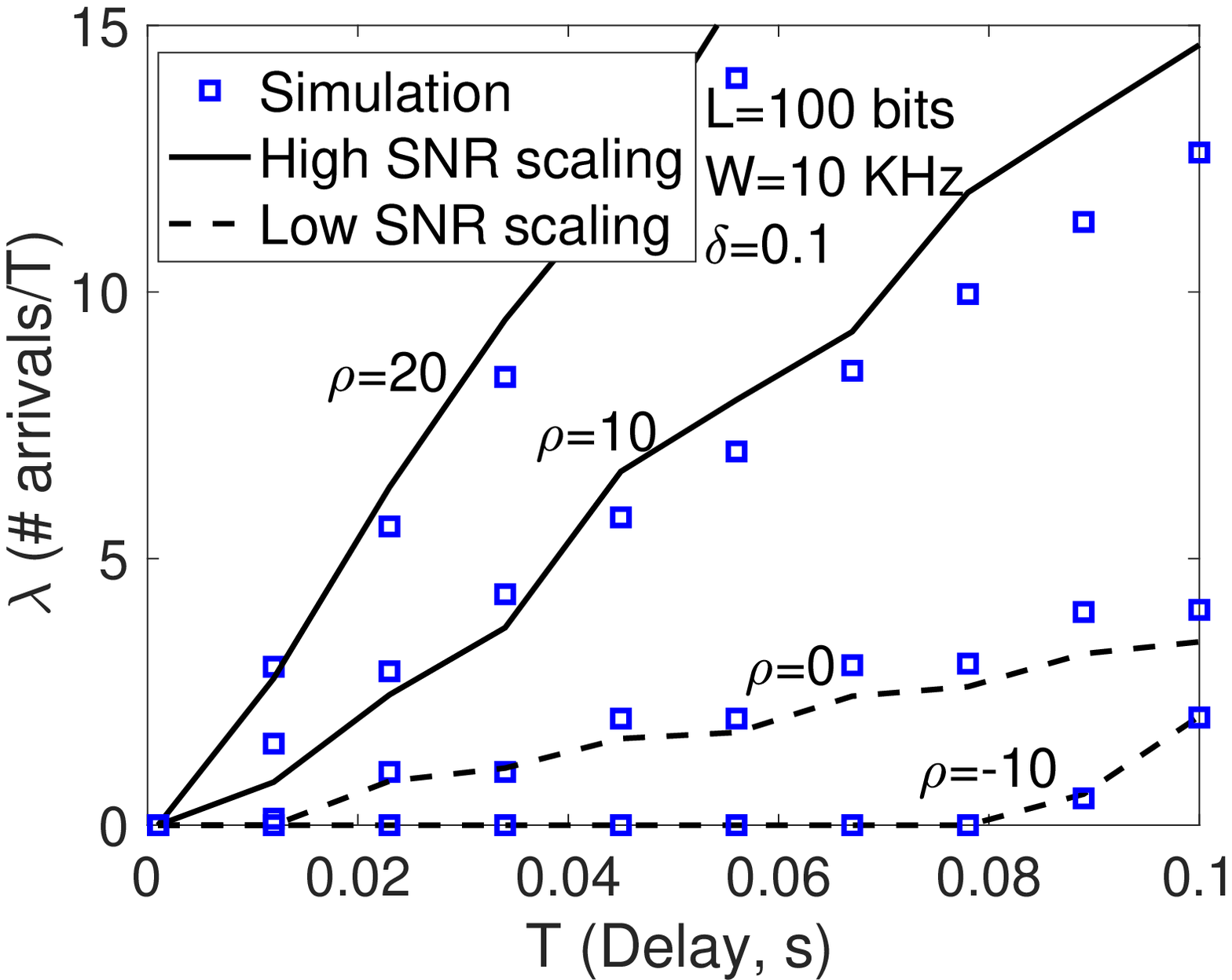}
\caption{\small{$\Lopt$ vs. delay constraint $T$ (s) at constant $\snr$ for different $\rho$ (dB) in the FBL regime.}\label{ThroughputvsDelayFBL}}
\end{minipage}
\end{figure*}

{\bf Throughput scaling with respect to latency constraint $T$.} At high $\snr$, the trend of $\Lopt$ versus $T$ 
is similar to trend between $\Lopt$ and $W$. At low $\snr$, we observe that $\kopt$ changes when $T$ is a multiple of $\frac{L}{W}\log(2)$. Hence, it may be approximated as linear for sufficiently large $W$. We plot the variation of $\Lopt$ with $T$ in Fig. \ref{ThroughputvsDelay}-(a), for different $\rho$ (dB). In the same figure, we also show the optimal values of $B$ and $M$ for a few points along the curves. In Fig. \ref{Moptvsdelta}, we illustrate $\Mopt$, which is a non-increasing function of $\delta$, for different $\rho$ (dB) at constant $\snr$. 

Although the variability of the channel causes a drop in the number of arrivals that can successfully complete the random access phase, in Fig. \ref{ThroughputvsDelay}-(b) we showed that similar conclusions in the constant $\snr$ case extend to the case of Rayleigh fading. 

Finally, we consider the $\Lopt$ versus $T$ trend of the FBL regime at constant $\snr$, and compare it to the IBL case, which is shown in Fig. \ref{ThroughputvsDelayFBL}. For a given $L$, $T$, and $W$, the IBL throughput is an upper bound to the FBL throughput. This is mainly due to the stricter capacity constraint of the FBL model. 
Note from (\ref{dispersion}) that $V(\SINR)$ increases in $\sinr$. Therefore, from (\ref{FBLrate}), we infer that given a block length $n$ and block error rate $\varepsilon$, the gap between the maximum achievable rate for the FBL regime and the Shannon capacity increases in $\rho$. 


\section{Conclusions}
\label{conc}
We proposed a random access model for a single cell uplink where the objective is to characterize the fundamental limits of supported arrival rate by optimizing the number of retransmissions and resource allocation under a maximum latency constraint. We evaluate the performance of the model with respect to total bandwidth, latency constraint, payload, outage constraint, and $\snr$. We obtain the following design insights for the  infinite block length (IBL) and finite block length (FBL) models for different $\snr$, with different packet sizes and outage constraints:
\begin{itemize}
\item Independent of the payload size $L$, at low $\snr$, the resources should be shared, implying small $\Bopt$, $\Mopt$. As $\snr$ increases, the resources should be split, yielding large $\Bopt$, $\Mopt$.
\item As the target $\sinr$ outage rate $\delta$ 
decreases, $\Bopt$ and $\Mopt$ increase.
\end{itemize}

The insights can be applied to the design of narrowband IoT systems for 5G with ultra-reliability, ultra-low latency, and a large number of connected devices. Possible extensions include the modeling and analysis of the random access in a cellular network setting, by capturing the scaling of throughput using realistic power control and fading models, and characterizing the uplink interference. They also include the analysis of heterogeneous random access 
to accommodate a wide range of device-chosen strategies, 
capabilities and requirements. The overhead caused by channel-estimation is significant and different diversity combining or multiplexing techniques can improve the system performance, especially when there is deep fading. Another direction would be to incorporate the FBL model in a cellular network setting, and characterize the finite block-error probability by capturing the uplink interference at low and high contention.

\begin{appendix}
\section{Appendices}
\subsection{Proof of Proposition \ref{chase}} \label{App:Appendix-chase}
In general, if the received signal vector during transmission $i\in\mathcal{M}$ is $\mathbf{r}_{il}=a_{il} \mathbf{s}_l+\mathbf{n}_i+\mathbf{z}_{il}$, where $l\in\{1,\hdots, K_i\}$ is the index of the typical device, $\mathbf{s}_l\in \mathbb{C}^n$ is the desired signal, $a_{il}$ is the complex amplitude, $\mathbf{n}_i$ is an $n$-dimensional complex Gaussian noise vector with zero mean and covariance $\sigma_i^2\mathbf{I}_n$, and $\mathbf{z}_{il}=\sum\nolimits_{k=1, k\neq l}^{K_i} a_{ik} \mathbf{s}_k$ is the interference given that there are $K_i\geq 1$ devices transmitting in TFS $i$. Chase combining of $m$ transmissions results in
\begin{align*}
\sum\nolimits_{i=1}^m a_{il}^* \mathbf{r}_{il}
=\sum\nolimits_{i=1}^m |a_{il}|^2 \mathbf{s}_l   +  \sum\nolimits_{i=1}^m a_{il}^*( \mathbf{n}_i+\mathbf{z}_{il}).
\end{align*}
The signal power is $P_S=\big( \sum\nolimits_{i=1}^m |a_{il}|^2 \mathbf{s}_l \big)^2=\big(\sum\nolimits_{i=1}^m |a_{il}|^2\big)^2$. The noise power by TIN is
\begin{align*}
P_N&=\Big|\sum\nolimits_{i=1}^m  a_{il}^*( \mathbf{n}_i+\mathbf{z}_{il})\Big|^2=\sum\nolimits_{i=1}^m |a_{il}|^2 \sigma_i^2+\Big|\sum\nolimits_{i=1}^m a_{il}^*\mathbf{z}_{il}\Big|^2\\
&=\sum\nolimits_{i=1}^m |a_{il}|^2 \sigma_i^2+\Big|\sum\nolimits_{i=1}^m a_{il}^* \sum\nolimits_{k=1, k\neq l}^{K_i} a_{ik} \mathbf{s}_k\Big|^2\\
&=\sum\nolimits_{i=1}^m |a_{il}|^2 \sigma_i^2+\sum\nolimits_{i=1}^m \sum\nolimits_{k=1, k\neq l}^{K_i} |a_{il}|^2 |a_{ik}|^2.
\end{align*}
If the $\snr$ per user $|a_{il}|^2/\sigma_i^2=\rho$ and $\sigma_i^2=\sigma^2$ for all $i$ and $l$, then the $\snr$ at the output of the Chase combiner by treating interference as noise (TIN) as a result of $m$ transmissions is  
\begin{align}
\label{chaseSNR}
\SINR=\frac{\big(\sum\nolimits_{i=1}^m |a_{il}|^2\big)^2}{\sum\nolimits_{i=1}^m |a_{il}|^2 \sigma_i^2+\sum\nolimits_{i=1}^m \sum\nolimits_{k=1, k\neq l}^{K_i} |a_{il}|^2 |a_{ik}|^2}.
\end{align}

Letting $|a_{ik}|=\sqrt{ \rho\sigma^2}$ be the signal amplitude of each device $k=1,\hdots, K_i$ on the $i^{\rm th}$ attempt, the $\snr$ per device at the BS will be $\rho$, in a given resource bin. Then at constant $\snr$, the $\sinr$ during transmission $m$ is $\rho/(1+(K_i-1)\rho)$ if there are $K_i\geq 1$ devices arriving. Hence, the Chase combiner output $\sinr$ of the typical device as a result of $m\in\mathcal{M}$ transmissions, i.e. $\SINRm$ as a function of $\mathcal{K}_m$, $m\in\mathcal{M}$, is given as
\begin{align}
\SINRm=\frac{\rho m^2}{m+\rho \big(\sum\limits_{i=1}^m K_i-m\big)}, \quad m\in\mathcal{M}. \nonumber    
\end{align}

In the case of Rayleigh fading, let $|a_{il}|=\sqrt{ h_{l_i} \rho\sigma^2}$ be the received signal amplitude of each device at retransmission attempt $i$, where the typical device has index $l=1$, and $l\in\{2,\hdots, K_i\}$ is the interferer index, and $h_{l_i} \sim \exp(1)$ for all $l\in\{1,\hdots, K_i\}$. Hence, the Chase combiner output $\SINRm$, which is a function of the channel realizations and $\mathcal{K}_m$, $m\in\mathcal{M}$, is given as
\begin{align}
\SINRm=\frac{\rho\big(\sum\nolimits_{i=1}^m h_{1_i} \big)^2}{\sum\nolimits_{i=1}^m h_{1_i} (1+\sum\nolimits_{l=2}^{K_i}{\rho h_{l_i}})}, \quad m\in\mathcal{M}. \nonumber
\end{align}

\subsection{Proof of Proposition \ref{OutageGeneralM-nofading}}
\label{App:Appendix-OutageGeneralM-nofading}
The composite arrival process of the new devices plus the retransmission attempts in the steady state can be approximated by a homogeneous Poisson process with density $\lambda_M \geq \lambda$ \cite[Assumption 1]{Dhillon2014tcomm}. Therefore, the probability of outage given up to a maximum of $M$ attempts is
\begin{align*}
\PfailIBL(\lambda,L,B,M)=\sum\limits_{k_1=\kopt_1+1}^{\infty}\,\, \sum\limits_{k_2=\kopt_2(k_1)+1}^{\infty}\,\,\cdots\sum\limits_{\kopt_M=\kopt_M(\mathcal{K}_{M-1})+1}^{\infty}\,\, \prod\limits_{i=1}^{M} D\Big(k_i,\frac{\lambda_M}{B}\Big),
\end{align*}
where $\kopt_m$ is a function of $\mathcal{K}_{m-1}=\{k_1,\hdots, k_{m-1}\}$ for $2\leq m\leq M$ to be calculated based on the rate requirement. Note that $\PfailIBL(\lambda,L,B,m)$ monotonically decreases in $m$. 
The aggregate arrival rate $\lambda_M$ with up to $M$ total transmissions is $\lambda_M=\Lopt\big[1+\sum\nolimits_{m=1}^{M-1}\PfailIBL(\lambda,L,B,m) \big]$. 

\subsection{Proof of Proposition \ref{ShannonOutageFadingGeneralM}} 
\label{App:Appendix-ShannonOutageFadingM}
The outage probability given up to a maximum of $M$ transmissions with Rayleigh fading is 
\begin{align}
\label{OutageProbFadingGeneralM}
\PfailIBL(\lambda,L,B,M)
&=\mathbb{E}\Big[\Big.\mathbb{P}\Big(\underset{m\in\mathcal{M}}{\max}\{C(\SINRm)\}<\frac{L}{n}\Big\vert \mathcal{K}_M\Big)\Big]\nonumber\\
&\stackrel{(a)}{=}\mathbb{E}_{\mathcal{K}_M}\Big[\mathbb{E}_{I_{\mathcal{K}_M}}\Big[\mathbb{P}\Big(h<\frac{\Gamma}{\rho m^2}\big(m+ \sum\limits_{i=1}^m{I_{K_i}}\big)\Big\vert I_{\mathcal{K}_M}, \mathcal{K}_M\Big)\Big]\Big]\nonumber
\end{align}
\begin{align}
&\stackrel{(b)}{=}\mathbb{E}_{\mathcal{K}_M}\Big[\mathbb{E}_{I_{\mathcal{K}_M}}\Big[\prod\nolimits_{m=1}^M 1-e^{-\frac{\mu\Gamma}{\rho m^2}\big(m+ \sum\limits_{i=1}^mI_{K_i}\big)}\Big\vert I_{\mathcal{K}_M}, \mathcal{K}_M\Big]\Big]\nonumber\\
&=1-e^{-\mu\Gamma\rho^{-1}}\mathbb{E}_{K_1}\left[\left. \mathcal{L}_{I_{K_1}}(\mu\Gamma\rho^{-1})\right\vert K_1\right]
-e^{-2\mu\Gamma\rho^{-1}}\mathbb{E}_{\mathcal{K}_2}\left[\left. \mathcal{L}_{I_{K_1}}(\mu\Gamma\rho^{-1})\mathcal{L}_{I_{K_2}}(\mu\Gamma\rho^{-1})\right\vert \mathcal{K}_2\right]\nonumber\\
&+e^{-3\mu\Gamma\rho^{-1}}\mathbb{E}_{\mathcal{K}_2}\left[\left. \mathcal{L}_{I_{K_1}}(2\mu\Gamma\rho^{-1})\mathcal{L}_{I_{K_2}}(\mu\Gamma\rho^{-1})\right\vert \mathcal{K}_2\right]
+\hdots\nonumber\\
&\stackrel{(c)}{=}\sum\limits_{l_1=0}^1\,\sum\limits_{l_2=l_1}^{l_1+1}\,\hdots \sum\limits_{l_M=l_{M-1}}^{l_{M-1}+1} (-1)^{l_M}\prod\limits_{i=1}^{M} f\Big(\mu,\frac{\lambda_M}{B},l_i\Gamma\Big),
\end{align} 
where we used the shorthand notation for interference as $I_{\mathcal{K}_m}\coloneqq\{I_{K_1}, \hdots, I_{K_m}\}$, and $(a)$ follows from the definition of $\sinr$ in (\ref{SINRdefnRayleighFading}) and $(b)$ from that $h\sim \exp(\mu)$, and $(c)$ from the definition of $D\big(k,\frac{\lambda_M}{B}\big)$, independence of $I_{K_i}$'s, and $\mathcal{L}_{I_{K_i}}(s)=\mathcal{L}_{I_{K_1}}(s)$ for $i\in \mathcal{M}$. Using the Laplace transform of the interference given $k$ arrivals, which is denoted by $\mathcal{L}_{I_{k}}(s)$ and expressed as
\begin{align}
\mathcal{L}_{I_{k}}(s)=\mathbb{E}[e^{-sI_{k}}]=\mathbb{E}\Big[e^{-s\rho\sum\limits_{i=1}^{k-1}g_{i_m}}\Big]=\mathbb{E}\Big[\prod\limits_{i=1}^{k-1}e^{-s\rho g_{i}}\Big]=\Big(\frac{\mu}{\mu+s\rho}\Big)^{k-1}.\nonumber
\end{align}
Averaging $\mathcal{L}_{I_{K}}(s)$ using the distribution $D(K,\lambda_B)$ given in (\ref{Dfunction}), we obtain 
\begin{align}
\label{averageLTofICI}
\mathcal{L}_{\ICI}(s)=\mathbb{E}_{K}[\mathcal{L}_{I_{K}}(s)]=\sum\limits_{k=1}^{\infty}{D(k,\lambda_B)\Big(\frac{\mu}{\mu+s\rho}\Big)^{k-1}}=(s\mu^{-1}\rho+1) \frac{e^{\lambda_B/(s\mu^{-1}\rho+1)}-1}{e^{\lambda_B}-1}.
\end{align}
To simplify the notation, we let $s=\mu\Gamma\rho^{-1}$ and define $f(\mu,\alpha,\Gamma)=e^{-\mu\Gamma\rho^{-1}}(\Gamma+1) \frac{e^{\alpha/(\Gamma+1)}-1}{e^{\alpha}-1}$, using which the final expression $(c)$ is evaluated.

\subsection{Proof of Proposition \ref{Bopt_generalM_nofading-HighSNR}}
\label{App:Appendix-Bopt_generalM_nofading-HighSNR}
At high $\snr$, since $\kopt_1=1$, 
$\SINR=\rho$. Since the total number of resources $N$ is split into $M$ retransmissions and $B$ bins to successfully transmit $L$ bits, in the IBL regime, we have $B\approx \frac{N}{ML}\log_2(1+\rho)$. 
The relation between the target outage rate $\delta$ and $\alpha_M(\delta)=\lambda_M/B$ is  
\begin{align*}
\PfailIBL(\lambda,L,B,M)=\delta=\Big(1-\frac{\alpha_M(\delta)}{e^{\alpha_M(\delta)}-1}\Big)^M=\PfailIBL(\lambda,L,B,m)^{\frac{M}{m}},\,\, m\in\mathcal{M},
\end{align*}
where $\alpha_{M}(\delta)=\alpha_1(\delta^{1/M})=\alpha(\delta^{1/M})$. The relation between $\lambda$ and the aggregate arrival rate is 
\begin{align*}
\lambda
&=\lambda_M\Big/\Big[1+\sum\limits_{m=1}^{M-1}\delta^{\frac{m}{M}}\Big]=\lambda_M(1-\delta^{1/M})/(1-\delta).
\end{align*}
The optimal number of transmissions can be found solving
$\Mopt=\underset{M \geq 1}{\arg\max}\,\, \lambda_M(1-\delta^{1/M})$, 
which is also a function of $\delta$ at high $\snr$, and $B$ is given as $B 
\approx \frac{N}{ML}\log_2(1+\rho)$.

\subsection{Proof of Proposition \ref{lowSNRsufficientcondition}}
\label{App:Appendix-lowSNRsufficientcondition}
The condition $\rho <2^{2L/N}-1$ can be rewritten as  
$\frac{N}{L}\log_2(1+\rho)<2$.  
Given $M$, the capacity constraint in (\ref{CapacityConstraintGeneralModel}) is 
$C(\SINRm)\geq \frac{L}{n},\,\,m\in\mathcal{M},\,\, n\leq \frac{N}{MB}$. Equivalently, we have $\frac{N}{L}\log_2(1+\SINRm)\geq MB,\,\,m\in\mathcal{M}$. In addition, due to the interference, the Chase combiner output as a result of $m$ transmissions is upper bounded as $\rho m \geq \SINRm,
\,\, m\in\mathcal{M}$. 
Hence, given $m=1$, we can conclude that $\frac{N}{L}\log_2(1+\rho)\geq\frac{N}{L}\log_2(1+\SINRone)\geq MB$. Combining this with 
the $\snr$ condition, 
we have $MB<2$, which is a sufficient condition for $\Bopt=\Mopt=1$.

\subsection{Proof of Proposition \ref{Lambda_approx_low_SNR}} 
\label{App:Appendix-Lambda_approx_low_SNR}
The proof follows from approximating the 
device arrival rate 
per bin $\lambda/B$ 
using the Gaussian distribution that has the same mean $\lambda/B$ and variance $\lambda/B$ as the Poisson distribution, i.e., $\mathcal{N}\big(\lambda/B,\lambda/B\big)$ for sufficiently large values of $\lambda/B$, ($\lambda/B>1000$). At low $\snr$, as $k_{\max}\to \infty$, we have $\lambda/B \to \infty$. Therefore, the normal distribution is an excellent approximation to the Poisson distribution. Hence, given a target $\sinr$ outage rate $\delta$ per device, the probability of outage is approximated as $\delta=D\big(k_1,\frac{\lambda}{B}\big) \approx \Q\Big(\frac{\kopt_1-\lambda/B}{\sqrt{\lambda/B}}\Big)$, where $\Q(x)=\frac{1}{\sqrt{2\pi}}\int\nolimits_{x}^{\infty}{e^{-u^2/2}{\rm d}u}$ is the tail probability of the standard normal distribution. Therefore, 
$\sqrt{\frac{\lambda}{B}}\approx \sqrt{\kopt_1+\frac{(\Q^{-1}(\delta))^2}{4}}-\frac{\Q^{-1}(\delta)}{2}$, where $\Q^{-1}$ is the inverse $\Q$-function, and $\kopt_1\leq \frac{1}{\Gamma}-\frac{1}{\rho}+1$. Thus, we can infer that $\lambda/B$ scales linearly with $\kopt_1$. However, since $B\leq \frac{N}{L}\log_2\Big(1+\frac{\rho}{1+\rho(\kopt_1-1)}\Big)$, the number of bins, $B$, decays sub-linearly with $\kopt_1$. Thus, their product $\lambda$ increases as $\kopt_1$ increases. 
 
Because we cannot achieve arbitrarily high $\lambda/B$ picking $B$ arbitrarily small, and $B\in\mathbb{Z}^+$, we require $B=1$. The above approximation with $B=1$ gives the maximum supportable rate $\Lopt$ at low $\snr$. Hence, given a target $\sinr$ outage rate $\delta$, we evaluate the probability of outage for $\Bopt=\Mopt=1$ using (\ref{PoutGeneralShannon}) and $\kopt_m(\mathcal{K}_{m-1})$ using (\ref{jmaxgeneralM}). Then, we want to determine the optimal value of $\lambda$ that satisfies 
$\delta=\sum\limits_{k_1>\kopt_1}D\big(k_1,\lambda\big)$, 
where $\lambda$ is increasing in $\kopt_1$. 
Hence, we can approximate the maximum 
arrival rate $\Lopt$ using
\begin{align}
\label{deltaQapprox}
\delta\approx\Q\big(\big.{(\kopt_1-\Lopt)}\big/{\sqrt{\Lopt}}\big),\quad \kopt_1=\Big\lfloor \frac{1}{\Gamma}-\frac{1}{\rho}+1\Big\rfloor,
\end{align} 
where $\Gamma=2^{\frac{L}{N}}-1$. Thus, at low $\snr$, if $\rho$ is in a range such that $\Big(1+\frac{1}{\Gamma}-\kopt_1\Big)^{-1} \leq \rho <  \Big(\frac{1}{\Gamma}-\kopt_1\Big)^{-1}$, it yields $\Bopt=1$ and $\Lopt$ that satisfies (\ref{deltaQapprox}. 	       
Using (\ref{deltaQapprox}), the maximum average arrival rate that can be supported at low $\snr$ can be determined.

\subsection{Proof of Proposition \ref{FBLoutagenofadinggeneralM}}
\label{App:AppendixFBLoutagenofadinggeneralM}
The probability of outage for the FBL model at constant $\snr$ with $M$ retransmissions equals
\begin{align}
\PfailFBL(\lambda,L,B,M)
&=\mathbb{E}\Big[\left.\mathbb{P}\Big(\underset{m\in\mathcal{M}}{\max}\{\Cfbl_{n,\varepsilon}(\SINRm)\}<\frac{L}{n}\right\vert \mathcal{K}_M\Big)\Big]\nonumber\\
&\stackrel{(a)}{=}\mathbb{E}\Big[\Big. \mathbb{P}\Big(\frac{n\Cfbl_{n,\varepsilon}(\SINRm)-L}{\sqrt{nV(\SINRm)}}<0, m\in\mathcal{M}\Big\vert \mathcal{K}_M \Big)\Big]\nonumber\\
&\stackrel{(b)}{=}\mathbb{E}\left[\Big.\mathbb{P}\Big(\mathcal{Z}<-f(n,L,\SINRm), m\in\mathcal{M}\Big\vert \mathcal{K}_M \Big)\right]\nonumber\\
&\stackrel{(c)}\approx\sum\limits_{k_1=1}^{\infty}\, \cdots \,\sum\limits_{k_M=1}^{\infty} \,\prod\limits_{i=1}^M D\Big(k_i,\frac{\lambda_M}{B}\Big)\Q(f(n,L,\sinri)),\nonumber
\end{align}
where $(a)$ follows from the definition of $\Cfbl$ using (\ref{FBLrate}) and (\ref{dispersion}), $(b)$ from the normal distribution approximation such that $\mathcal{Z}\sim
\mathcal{N}(0,1)$, and $(c)$ from the approximation based on standard Gaussian ccdf $\Q(x)$ and the relation $\Q(x)=1-\Q(-x)$, and the final result from (\ref{BlockErrorRate}).

\subsection{Proof of Lemma \ref{fconcaveofx}}
\label{App:Appendixfconcaveofx}
We rewrite $f(n,L,X)$ as $f(n,L,X)=\frac{C(X)+b}{\sqrt{V(X)/n}}$, and let $a=\frac{\sqrt{n}}{\log_2e}$, $b=\frac{0.5\log_2n-L}{n}$, $f_1(x)=\frac{(1+x)}{\sqrt{x(x+2)}}$ and $f_2(x)=\log_2(1+x)$. Then, we can express $f(n,L,x)$ as $f(n,L,x)=af_1(x)[f_2(x)+b]$. Then, we can derive the identities $\frac{\partial}{\partial x}f_1(x)=-\frac{1}{(x(x+2))^{3/2}}$, $\frac{\partial^2 }{\partial x^2}f_1(x)=\frac{3(1+x)}{(x(x+2))^{5/2}}$, $\frac{\partial}{\partial x}f_2(x)=\frac{1}{\log(2)}\frac{1}{1+x}$ and $\frac{\partial^2}{\partial x^2}f_2(x)=-\frac{1}{\log(2)}\frac{1}{(1+x)^2}$. We also assume that $b\leq 0$. From $b=\frac{0.5\log_2n-L}{n}\leq 0$ that is true when $n\leq 2^{2L}$, which is a realistic assumption for the FBL regime.

\begin{enumerate}[(i)]
\item {\bf Monotonically increasing.} 
Take the partial derivative of $f(n,L,x)$ with respect to $x$:
\begin{align}
\frac{\partial f(n,L,x)}{\partial x}&=a\left[\frac{\partial}{\partial x}f_1(x)[f_2(x)+b]+f_1(x)\frac{\partial}{\partial x}f_2(x)\right]\nonumber\\
&=\sqrt{n}\frac{1}{(x^2+2x)^{1/2}}\left[1-\frac{\log_2(1+x)+b}{\log_2(e)(x^2+2x)}\right]\stackrel{(a)}{>}0,\,\,x\geq 0.\nonumber
\end{align}
where $(a)$ is due to $x\geq 0$, $1-\frac{\log(1+x)}{x^2+2x}>0$, and  $b\leq 0$. Thus, $f(n,L,x)$ is increasing in $x$.

Note that $f_1(x)$ and $f_2(x)$ are increasing in $x\geq 0$. Hence, $f(n,L,x)$ is also increasing in $x$. Let $x\leq y$. Then, $f(n,L,x)=af_1(x)[f_2(x)+b]\leq af_1(y)[f_2(y)+b]=f(n,L,y)$. Hence, $f(n,L,x)$ is a monotonically increasing function of $x$.

\item {\bf Positive.}
From definition of $\Cfbl$, it is required to satisfy
\begin{align}
\Cfbl_{n,\varepsilon}(\SINR)= \log_2(1+\SINR)-\sqrt{\frac{V(\SINR)}{n}}\Q^{-1}(\varepsilon)+\frac{0.5\log_2(n)}{n}+o\left(\frac{1}{n}\right)\geq \frac{L}{n},\nonumber
\end{align}
which yields the condition that $f(n,L,\SINR)=\frac{\log_2(1+X)+b}{\sqrt{V(X)/n}}\geq \Q^{-1}(\varepsilon)$. Note that since $\Q^{-1}(\varepsilon)>0$ as long as $\varepsilon<0.5$. This condition guaranties the positivity of $f(n,L,x)$.

\item {\bf Concave.} To prove that $f(n,L,x)$ is concave, derive its second partial derivative 
as
\begin{align*}
\frac{\partial^2 }{\partial x^2}f(n,L,x)&=a\left[\frac{\partial^2 }{\partial x^2}f_1(x)[f_2(x)+b]+2\frac{\partial}{\partial x}f_1(x)\frac{\partial}{\partial x}f_2(x)+f_1(x)\frac{\partial^2 }{\partial x^2}f_2(x)\right]\\
&\stackrel{(a)}{=}\sqrt{n}\frac{x+1}{(x^2+2x)^{3/2}}\left[\frac{3[\log_2(1+x)+b]}{\log_2(e)(x^2+2x)}-1-\frac{1}{(1+x)^2}\right]<0,\,\, x\geq 0,
\end{align*}
where $(a)$ follows from that $\frac{3\log_2(1+x)}{\log_2(e)(x^2+2x)}-1-\frac{1}{(1+x)^2}<0$ for $x\geq 0$ and from $b\leq 0$.
\end{enumerate}

\subsection{Proof of Proposition \ref{BoundsonQfunction}}
\label{App:AppendixBoundsonQfunction}
Since the second derivative of $\Q(x)$ is always greater than or equal to zero, $\Q(x)$ is convex and increasing in $x\geq 0$. This implies that for a random variable $X$, we have $\mathbb{E}[\Q(X)]\geq \Q(\mathbb{E}[X])$.

From Lemma \ref{fconcaveofx}, $f(n,L,X)$ is monotonically increasing, positive, and concave for $X\geq 0$ and $\Q(f(n,L,X))$ is decreasing in $X$. Then, $\mathbb{E}[f(n,L,X)]\leq f(n,L,\mathbb{E}[X])$. Note also that $\Q(X)$ is decreasing, monotone, and convex for $X>0$. Thus, $\mathbb{E}[\Q(X)]\geq \Q(\mathbb{E}[X])$. Then, using that $(a)$ $\Q$ is convex, and $(b)$ $f(n,L,X)$ is concave and $\Q$ is decreasing in $X$, we have that
\begin{align}
\mathbb{E}[\Q(f(n,L,X))]\overset{(a)}{\geq} \Q(\mathbb{E}[f(n,L,X)]) \overset{(b)}{\geq} \Q(f(n,L,\mathbb{E}[X])). \nonumber
\end{align}

\subsection{Proof of Proposition \ref{HighSNRFBLcoverage}}
\label{App:Appendix-HighSNRFBLcoverage}
At high $\snr$, the resource symbols are 
shared such that $k_m=1$, $m\in\mathcal{M}$. Hence, from (\ref{PoutdefinitionFBL}), 
\begin{align}
\PfailFBL(\lambda,L,B,M)=\prod\limits_{m=1}^M D\Big(1,\frac{\lambda_M}{B}\Big)\eber(n,L,\rho m)\approx D\Big(1,\frac{\lambda_M}{B}\Big)^M\prod\limits_{m=1}^M \Q(f(n,L,\rho m)).\nonumber
\end{align}

\subsection{Proof of Proposition \ref{LowSNRFBLcoverageBounds}}
\label{App:Appendix-LowSNRFBLcoverageBounds}
At low $\snr$, for the FBL regime, using (\ref{JensenonQAndDecreasingQ}) and (\ref{OutageLowSNRFBL}), the probability of outage can be bounded as 
$\PfailFBL(\lambda,L,B,1)\geq \Q(f(n,L,\mathbb{E}_{\mathcal{K}}[\SINRone]))$, 
where the mean $\SINR$ is given by
\begin{align}
\mathbb{E}_{\mathcal{K}}[\SINRone]=\sum\limits_{k=1}^{\infty}D\Big(k,\frac{\lambda}{B}\Big)\frac{\rho}{1+\rho (k-1)}
=\frac{\rho}{1-\rho}\left[1+\frac{\big(-\frac{\lambda}{B}\big)^{1-\frac{1}{\rho}}}{e^{\lambda/B}-1}\Big(\GF(\rho^{-1})-\GF\Big(\rho^{-1},-\frac{\lambda}{B}\Big)\Big)\right],\nonumber
\end{align}
where $\GF(s)=\int _{0}^{\infty }t^{s-1}\,e^{-t}\,{\rm {d}}t$ is the Gamma function, and $\GF(s,x)=\int _{x}^{\infty }t^{s-1}\,e^{-t}\,{\rm {d}}t$ is the upper incomplete gamma function. Using the Chernoff bound of the $\Q$-function, which is $\Q(x)\leq e^{-x^2/2}$ for $x>0$, we can obtain an upper bound for the probability of outage as follows:
\begin{align}
\PfailFBL(\lambda,L,B,1)\leq \mathbb{E}_{\mathcal{K}}\left[\exp{\Big(-\frac{1}{2}f(n,L,\SINRone)^2\Big)}\right]
=\sum\limits_{k=1}^{\infty}D\Big(k,\frac{\lambda}{B}\Big)\exp{(-f(n,L,\sinrone)^2)}.\nonumber
\end{align}

\end{appendix}

\begin{spacing}{1.36}
\bibliographystyle{IEEEtran}
\bibliography{M2MLowLatencyreferences}
\end{spacing}

\end{document}